\documentclass[11pt]{article}

\usepackage[T1]{fontenc}
\usepackage{nameref}
\usepackage[colorlinks=true, urlcolor=blue, citecolor=blue, linkcolor=blue]{hyperref}
\usepackage{amsfonts, amsmath, amssymb, amsthm, thmtools}
\usepackage{mathrsfs}
\usepackage{cleveref}

\declaretheorem[name=Theorem, refname={Theorem, Theorems}, Refname={Theorem, Theorems}, parent=section]{theorem}
\declaretheorem[name=Assumption, refname={Assumption, Assumptions}, Refname={Assumption, Assumptions}, sibling=theorem, style=definition]{assumption}
\declaretheorem[name=Corollary, refname={Corollary, Corollaries}, Refname={Corollary, Corollaries}, sibling=theorem]{corollary}
\declaretheorem[name=Definition, refname={Definition, Definitions}, Refname={Definition, Definitions}, sibling=theorem, style=definition]{definition}
\declaretheorem[name=Example, refname={Example, Examples}, Refname={Example, Examples}, sibling=theorem, style=definition]{example}

\declaretheorem[name=Lemma, refname={Lemma, Lemmas}, Refname={Lemma, Lemmas}, sibling=theorem]{lemma}

\newtheorem{remark}{Remark}

\crefname{section}{section}{sections}
\Crefname{section}{Section}{Sections}
\crefname{table}{table}{tables}
\Crefname{table}{Table}{Tables}

\usepackage{color, enumitem, float, fourier, layout, multirow, setspace, verbatim}
\setlist[enumerate]{leftmargin=*}
\usepackage[colorinlistoftodos]{todonotes}
\usepackage{listings}
\usepackage{caption}
\usepackage{epigraph}
\usepackage{subfig}

\usepackage[round]{natbib}

\usepackage{graphicx}
\graphicspath{{figures/}{../figures/}}

\let\vec\mathbf

\def\I{{\mathbb I}}
\def\x{{\vec{x}}}
\def\X{{\vec{X}}}
\def\w{{\vec{w}}}
\def\y{{\vec{y}}}
\def\Y{{\vec{Y}}}

\def\E{{\mathbb E}}
\def\I{{\mathbb{I}}}

\def\P{{\mathbb P}}

\def\Re{{\mathbb R}}

\def\hF{{\widehat{F}}}

\def\nX{{\X_{n+1}}}
\def\nY{{Y_{n+1}}}

\def\sX{{\mathcal{X}}}
\def\sY{{\mathcal{Y}}}

\def\cdsplit{{\texttt{CD-split}}}

\def\distsplit{{\texttt{Dist-split}}}

\def\uceil{{U_{[\lceil 1-0.5 \alpha \rceil]}}}
\def\ufloor{{U_{[\lfloor0.5\alpha\rfloor]}}}
\def\hF{{\widehat{F}}}

\usepackage{makecell}

\usepackage{makecell}

\usepackage{amssymb}
\usepackage{pifont}
\newcommand{\cmark}{\ding{51}}%
\newcommand{\xmark}{\ding{55}}%

\usepackage{algorithm, algorithmic}
\renewcommand{\algorithmicrequire}{\textbf{\small Input:}}
\renewcommand{\algorithmicensure}{\textbf{\small Output:}}
\definecolor{darkgreen}{rgb}{0.3, 0.5, 0.0}
\newcommand{\codecomment}[1]{\textbf{\color{darkgreen}\ // #1}}

\usepackage{soul}

\usepackage{centernot}

\title{Flexible distribution-free
	conditional predictive bands \\ 
	using density estimators}

\author{%
	Rafael Izbicki \\
	University of S\~{a}o Carlos \\
	\texttt{rafaelizbicki@gmail.com} \\
	\\
	Gilson Y. Shimizu \\
	University of S\~{a}o Carlos \\
	S\~{a}o Paulo\\
	University of \texttt{gilsonshimizu@yahoo.com.br} \\
	\\
	Rafael B. Stern \\
	University of S\~{a}o Carlos \\
	\texttt{rbstern@gmail.com} \\  
	\\
	\\
}

\usepackage{amsmath}

\usepackage{setspace} 
\def\hF{{\widehat{F}}}
\def\hf{{\widehat{f}}}
\def\uceil{{U_{\lceil 1-0.5 \alpha \rceil}}}
\def\ufloor{{U_{\lfloor0.5\alpha\rfloor}}}
\def\hf{{\widehat{f}}}
\def\hF{{\widehat{F}}}

\usepackage{mathtools}
\DeclarePairedDelimiter{\ceil}{\lceil}{\rceil}

\begin{document}

\maketitle

\begin{abstract}
 Conformal methods create prediction bands that 
 control average coverage under no assumptions 
 besides i.i.d. data.
 Besides average coverage, one might
 also desire to control conditional coverage, 
 that is, coverage for every new testing point.
 However, without strong assumptions,
 conditional coverage is unachievable. 
 Given this limitation, the literature has
 focused on methods with asymptotical
 conditional coverage.
 In order to obtain this property, 
 these methods require strong conditions
 on the dependence between
 the target variable and the features.
 We introduce two conformal methods 
 based on conditional density estimators
 that do not depend on this type of assumption
 to obtain asymptotic conditional coverage: 
 Dist-split and CD-split.  
 While Dist-split asymptotically 
 obtains optimal intervals, 
 which are easier to interpret than general regions, 
 CD-split obtains optimal size regions,
 which are smaller than intervals.
 CD-split also obtains local coverage by 
 creating a data-driven partition of
 the feature space 
 that scales to high-dimensional settings  
 and by generating
 prediction bands locally 
 on the partition elements. 
 In a wide variety of 
 simulated scenarios,
 our methods have a better
 control of conditional coverage and
 have smaller length than
 previously proposed methods.
\end{abstract}

\section{Introduction}

Supervised machine learning methods predict
a response variable, $Y \in \sY$,
based on features, $\X \in \sX$, using
an i.i.d. sample, $(\X_1,Y_1),\ldots,(\X_n,Y_n)$. 
While most methods yield point estimates,
it is often more informative to
present prediction bands, that is,
a subset of $\sY$ with 
plausible values for $Y$ \citep{Neter1996}. 

A particular way of constructing 
prediction bands is through
\emph{conformal predictions}
\citep{Vovk2005,Vovk2009}.
This methodology is appealing because it controls 
the {\em marginal coverage} of
the prediction bands assuming
solely i.i.d. data.
Specifically, given a new instance,
$(\X_{n+1},Y_{n+1})$, a conformal prediction,
$C(\X_{n+1})$, satisfies
\begin{align*}
 \P\left(Y_{n+1} \in C(\X_{n+1}) \right) 
 &\geq 1-\alpha,
\end{align*}
where $0<1-\alpha<1$ is a desired coverage level.
Besides marginal validity one might also 
wish for stronger guarantees.
For instance, \emph{conditional validity} holds
when, for every $\x_{n+1} \in \sX$,
\begin{align*}
 \P(Y_{n+1}\in C(\textbf{X}_{n+1})|\textbf{X}_{n+1}=\textbf{x}_{n+1})
 &\geq 1-\alpha.
\end{align*}
That is, conditional validity guarantees
adequate coverage for each new instance and
not solely on average across instances.

Unfortunately, conditional validity can be
obtained only under strong assumptions about the the distribution of $(\X,Y)$
\citep{Vovk2012,Lei2014,Barber2019}.
Given this result, effort has been focused
on obtaining  intermediate conditions. For instance, many conformal methods
control \emph{local coverage}:
\begin{align*}
 \P(Y_{n+1}\in C(\textbf{X}_{n+1})|\textbf{X}_{n+1}\in A)
 &\geq 1-\alpha,
\end{align*}
where $A$ is a subset of $\mathcal{X}$
\citep{Lei2014,Barber2019,Guan2019}.
These methods are based on computing
conformal bands using only 
instances that fall in $A$.
However, to date, these methods do not 
scale to high-dimensional settings because 
it is challenging to create $A$ that 
is large enough so that many instances fall in $A$,
and yet small enough so that
\begin{align*}
 \P(Y_{n+1}\in C(\textbf{X}_{n+1})|\textbf{X}_{n+1}\in A) \approx 
 \P(Y_{n+1}\in C(\textbf{X}_{n+1})|\textbf{X}_{n+1}=\x_{n+1}),
\end{align*}
that is, local validity is close to conditional validity.
 
Another alternative to conditional validity is 
\emph{asymptotic} conditional coverage \citep{Lei2018}.
Under this property, conditional coverage
converges to the specified level as
the sample size increases. That is,
there exist random sets, $\Lambda_n$, such that 
$\P(X_{n+1} \in \Lambda_n|\Lambda_n) = 1-o_P(1)$ and
\begin{align*}
 \sup_{x \in \Lambda_n}\big|\P(Y_{n+1} \in C(\textbf{X}_{n+1})
 |\textbf{X}_{n+1}=\x_{n+1}) - 1-\alpha \big| = o_P(1).
\end{align*}
In a regression context in which 
$\sY = \mathbb{R}$, \citet{Lei2018} obtains 
asymptotic conditional coverage under
assumptions such as  $Y = \mu(\X) + \epsilon$,
where  $\epsilon$ is independent of $\X$ and
has density symmetric around 0.
Furthermore, the proposed prediction band
converges to the interval with
the smallest interval among
the ones with adequate conditional coverage.

Despite the success of these methods,
there exists space for improvement.
In many problems the assumption that
$\epsilon$ is independent of $\X$ and
has a density symmetric around $0$ is
unrealistic. For instance,
in heteroscedastic settings \citep{Neter1996},
$\epsilon$ depends on $\X$.
It is also common for $\epsilon$  
to have an asymmetric or even multimodal distribution.
Furthermore, in these general settings,
the smallest region with adequate
conditional coverage might not be
an interval, which is
the outcome of most current methods.

\subsection{Contribution}

We propose new methods and show that
they obtain asymptotic conditional coverage
without assuming a particular type 
of dependence between the target and the features.
Specifically, we propose two methods:
\distsplit\ and \cdsplit .
While \distsplit\ produces prediction bands that
are intervals and easier to interpret, 
\cdsplit\ yields arbitrary regions,
which are generally smaller and appealing 
for multimodal data. 
While \distsplit\ converges to
an oracle interval, \cdsplit\
converges to an oracle region.
Furthermore, since \cdsplit\ is based on 
a novel data-driven way of 
partitioning the feature space,
it also controls local coverage even in 
high-dimensional settings. 
\Cref{tab:comparison_approaches} summarizes 
the properties of these methods.

\begin{table*}
 \centering
 \caption{Properties of \distsplit\ and \cdsplit.}
 \begin{tabular}{ c|c|c|c|c|c } 
  Method & \makecell{Marginal  \\ coverage}
  & \makecell{Asymptotic  \\conditional coverage}
  & \makecell{Local \\coverage}
  & \makecell{Prediction bands \\are intervals} 
  & \makecell{Can be used \\for classification?} \\
  \hline 
  \distsplit &\cmark &\cmark & \xmark & \cmark &\xmark\\
  \hline
  \cdsplit & \cmark & \cmark & \cmark &\xmark & \cmark
  \label{tab:comparison_approaches}
 \end{tabular}
\end{table*}

The proposed methods also have
desirable computational properties.
They are based on fast-to-compute 
split (inductive)-conformal bands 
\citep{Papadopoulos2008,Vovk2012,Lei2018} and
on novel conditional density estimation methods that
scale to high-dimensional datasets
\citep{Lueckmann2017,Papamakarios2017,Izbicki2017}
Both methods are easy to compute and
scale to large sample sizes as long as
the conditional density estimator also does.

In a wide variety of simulation studies,
we show that our proposed methods obtain
better conditional coverage and
smaller band length than
alternatives in the literature.
For example, \Cref{fig:example_coverage} 
illustrates  \cdsplit, \distsplit \ and
the reg-split method from \citet{Lei2018} 
on the toy example from \citet{Lei2014}. 
The bottom right plot shows that 
both \cdsplit \ and  \distsplit\ get close to
controlling conditional coverage. 
Since \distsplit\ can yield only intervals,
\cdsplit \ obtains smaller bands in
the region in which $\Y$ is bimodal.
In this region \cdsplit\ yields 
a collection of intervals around
each of the modes.

\begin{figure}
 \centering
 \subfloat{\includegraphics[width=0.25\textwidth]{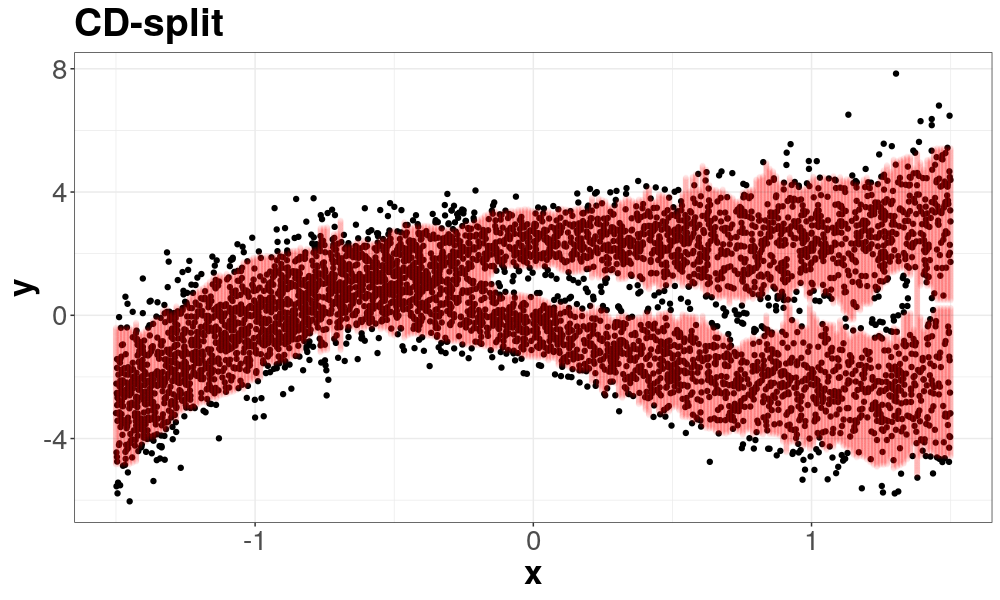}}
 \subfloat{\includegraphics[width=0.25\textwidth]{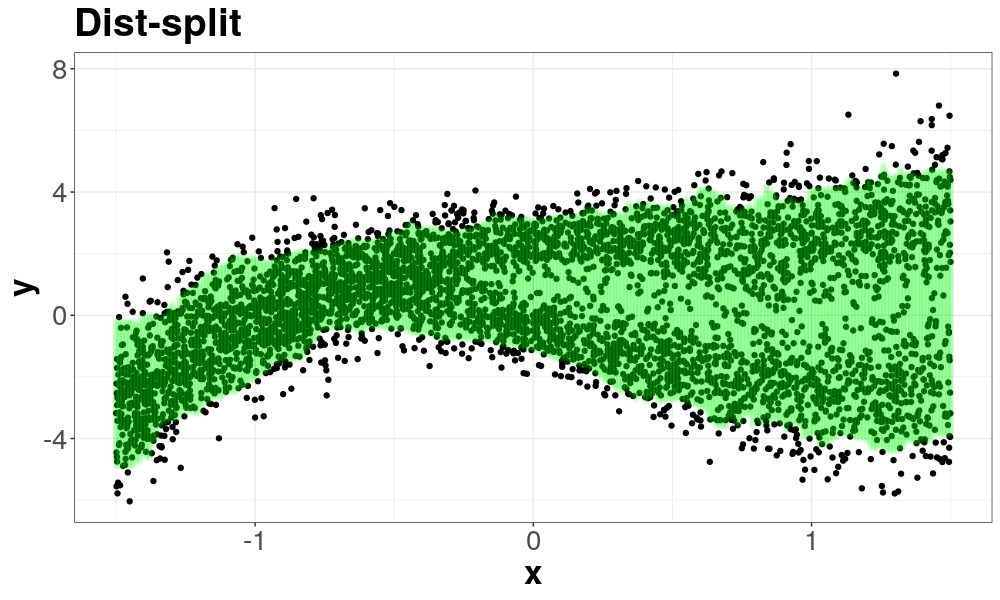}}
 \\
 \subfloat{\includegraphics[width=0.25\textwidth]{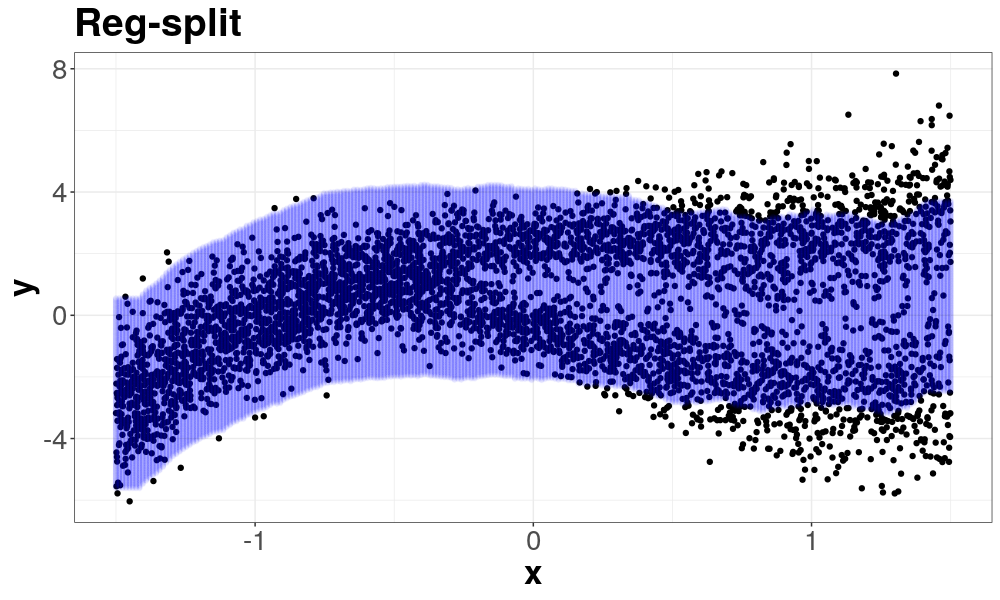}}
 \subfloat{\includegraphics[width=0.25\textwidth]{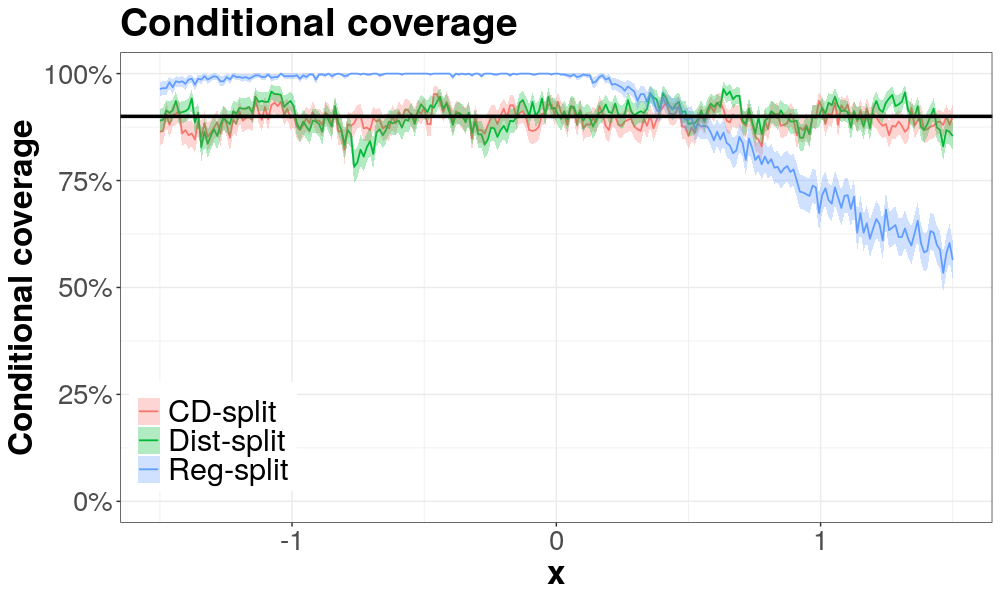}}
 \vspace{-2mm}
 \caption{Comparison between \cdsplit, \distsplit \ and
 the reg-split method from \citet{Lei2018}.} 
 \label{fig:example_coverage}
\end{figure}

The remaining of the paper is organized as follows.
Section \ref{sec:dist} presents \distsplit.
Section \ref{sec:cd} presents \cdsplit.
Experiments are shown in Section \ref{sec:exp}. Final remarks are on Section \ref{sec:final}.
All proofs are shown in the supplementary material.

\textbf{Notation.}
Unless stated otherwise, we study 
a univariate regression setting 
such that $\sY = \Re$.
Data from an i.i.d. sequence is
split into two parts,
$\mathbb{D}=\{(\X_1,Y_1),\ldots,(\X_n,Y_n)\}$
and $\mathbb{D}'=\{(\X'_1,Y'_1),\ldots,(\X'_n,Y'_n)\}$.
The assumption that both datasets have the same size
is used solely to simplify notation.
Also, the new instance, $(\X_{n+1},Y_{n+1})$, 
has the same distribution as the other sample units.
Finally, $q(\alpha;\{t_1,\ldots,t_m\})$ is
the $\alpha$ quantile of
$\{t_1,\ldots,t_m\}$.

\section{\distsplit}
\label{sec:dist}

The \distsplit \ method is based on
the fact that, if $F(y|\x)$ is
the conditional distribution of
$\nY$ given $\nX$, then
$F(\nY|\nX)$ has uniform distribution.
Therefore, if $\hF$ is close to $F$, then
$\hF(\nY|\nX)$ approximately uniform, 
and does not depend on $\nX$. That is,
obtaining marginal coverage for $\hF(\nY|\nX)$
is close to obtaining conditional coverage.

\begin{definition}[\distsplit\ prediction band]
 \label{def:dist_split}
 Let $\hF(y|\x_{n+1})$ be an estimate
 based on  $\mathbb{D}'$ of
 the conditional distribution of 
 $Y_{n+1}$ given $\x_{n+1}$.
 The \distsplit\ prediction band,
 $C\left(\x_{n+1}\right)$, is
 \begin{align*}
  \label{eq:cde_intervals}
  C&\left(\x_{n+1}\right)
  :=\left\{y: q(.5\alpha;\mathcal{T}(\mathbb{D}))
  \leq \hF(y|\x_{n+1}) \leq
  q(1-.5\alpha;\mathcal{T}(\mathbb{D}))\right\} \\
  &= \left[\hF^{-1}\left(q(.5\alpha;\mathcal{T}(\mathbb{D}))|\x_{n+1}\right);
  \hF^{-1}\left(q(1-.5\alpha;\mathcal{T}(\mathbb{D}))|\x_{n+1}\right) \right]
 \end{align*}
 where $\mathcal{T}(\mathbb{D})
 =\left\{\hF(Y_i|\X_i), i = 1,\ldots,n\right\}$. 
\end{definition}

Algorithm \ref{alg:dist} shows an
implementation of \distsplit.
\begin{algorithm}
 \caption{ \small \distsplit}\label{alg:dist}
 \algorithmicrequire \ {\small Data $(\textbf{X}_{i},Y_{i})$, $i=1,...,n$, 
 miscoverage level $\alpha \in (0,1)$, 
 algorithm $\mathcal{B}$ for fitting 
 conditional cumulative distribution function} \\
 \algorithmicensure \ {\small Prediction band for
 $\textbf{x}_{n+1}\in\mathbb{R}^d$}
 \begin{algorithmic}[1]
  \STATE Randomly split $\{1,2,...,n\}$ into two subsets $\mathbb{D}$ e $\mathbb{D}'$
  \STATE Fit $\hat{F} = \mathcal{B}(\{(\textbf{X}_i,Y_i):i \in \mathbb{D}' \})$ 
  \codecomment{Estimate cumulative distribution function}
  \STATE  Let $\mathcal{T}(\mathbb{D})=
  \{\hat{F}(y_i|\textbf{x}_i), i \in \mathbb{D}\}$
  \STATE Let $t_1=q(\alpha/2;\mathcal{T}(\mathbb{D}))$ and
  $t_2=q(1-\alpha/2;\mathcal{T}(\mathbb{D}))$
  \codecomment{Compute the quantiles of the set $\mathcal{T}(\mathbb{D})$}
  \STATE \textbf{return} $\left\{y:t_2  \geq \widehat{F}(y|\textbf{x}_{n+1})\geq t_1 \right\}$
 \end{algorithmic}
\end{algorithm}

\distsplit\ adequately controls
the marginal coverage.
Furthermore, it exceeds the
specified $1-\alpha$ coverage by
at most $(n+1)^{-1}$.
These results are presented in
\cref{thm:dist_split_control}.

\begin{theorem}[Marginal coverage]
 \label{thm:dist_split_control}
 Let $C(\textbf{X}_{n+1})$ be
 such as in \cref{def:dist_split}.
 If both $F(y|\x)$ and $\hF(y|\x)$ are
 continuous for every $\x \in \mathcal{X}$, 
 then
 \begin{align*}
  1-\alpha \leq \mathbb{P}(Y_{n+1} \in C(\textbf{X}_{n+1}))
  \leq 1-\alpha+\frac{1}{n+1}.
 \end{align*}
\end{theorem}

Under additional assumptions
\distsplit\ also obtains 
asymptotic conditional coverage and
converges to an optimal oracle band.
Two types of assumptions are required.
First, that the conditional density estimator,
$\hF$ is consistent. This assumption
is an adaptation to density estimators 
of the consistency assumption for
regression estimators in \citep{Lei2018}.
Also, we require that
$F(y|\x)$ is differentiable and
$F^{-1}(\alpha^*|\x)$ is 
uniformly smooth in a neighborhood of
$.5\alpha$ and $1-.5\alpha$.
These assumptions are formalized below.

\begin{assumption}[Consistency of density estimator]
 \label{ass:consistent_cde}
 There exist $\eta_n = o(1)$ and
 $\rho_n = o(1)$ such that
 \begin{align*}
  \P\left(\E\left[\sup_{y \in \sY}
  \left(\hF(y|\X)-F(y|\X)\right)^2
  \big|\hF\right]
  \geq \eta_n \right) \leq \rho_n
 \end{align*}
\end{assumption}

\begin{assumption}
 \label{ass:continuity}
 For every $\x \in \sX$,
 $F(y|\x)$ is differentiable.
 Also, if $q_{\alpha} = F^{-1}(\alpha)$,
 then there exists $M^{-1} > 0$ such that 
 $\inf_{\x}\frac{dF(y|\x)}{dy} \geq M^{-1}$
 in a neighborhood of
 $q_{0.5\alpha}$ and of 
 $q_{1-0.5\alpha}$.
\end{assumption}

Given the above assumptions,
\distsplit\ satisfies desirable
theoretical properties.
First, it obtains
asymptotic conditional coverage.
Also, \distsplit\ converges to
the optimal \emph{interval} according to
the commonly used \citep{parmigiani2009decision}
loss function
\begin{align*}
 L((a,b),Y_{n+1})
 &= \alpha (b-a) + (a-\nY)_{+} + (\nY-b)_{+},
\end{align*}
that is, \distsplit\ satisfies
\begin{align*}
 C(\nX) &\approx 
 \left[F^{-1}(.5\alpha|\nX); 
 F^{-1}(1-.5\alpha|\nX)\right]
\end{align*}
These results are formalized in
\cref{thm:dist_optimal}.

\begin{theorem}
 \label{thm:dist_optimal}
 Let $C_n(\X_{n+1})$ be the prediction band in
 \cref{def:dist_split} and
 $C^*(\X_{n+1})$ be the optimal prediction interval
 according to
 \begin{align*}
  L((a,b),Y_{n+1}) 
  &= \alpha (b-a)+(a-\nY)_{+} + (\nY-b)_{+}.
 \end{align*}
 Under \cref{ass:consistent_cde,ass:continuity},
 \begin{align*}
  \lambda(C_n(\X_{n+1}) \Delta C^*(\X_{n+1}))
  & = o_{\P}(1),
 \end{align*}
 where $\lambda$ is the Lebesgue measure.
\end{theorem}

\begin{corollary}
 Under \cref{ass:consistent_cde,ass:continuity}
 \distsplit\ achieves asymptotic
 conditional coverage.
\end{corollary}

\distsplit\ converges to
the same oracle as recently proposed
conformal quantile regression methods
\citep{Romano2019,Sesia2019}.
However, the experiments in
\Cref{sec:exp} show that \distsplit\
usually outperforms these methods.

If the distribution of $Y|\x$ is not
symmetric and unimodal, \distsplit\ may 
obtain larger regions than necessary.
For example, a union of two intervals better
represents a bimodal distribution than
a single interval.
The next section introduces \cdsplit\, 
which obtains prediction bands which
can be more general than intervals.

\section{\cdsplit}
\label{sec:cd}

The intervals that are output by \distsplit \
can be wider than necessary when
the target distribution is multimodal,
such as in \cref{fig:example_coverage}.
In order to overcome this issue,
\cdsplit\ yields prediction bands that
approximate $\left\{y:f(y|\x_{n+1}) > t\right\}$,
the highest posterior region.

A possible candidate for this approximation is
$\left\{y:\hf(y|\x_{n+1})>t\right\}$,
where $\hf$ is a conditional density estimator.
However, the value of $t$ that
guarantees conditional coverage varies
according to $\x$.
Thus, in order to obtain conditional validity,
it is necessary to choose $t$ adaptively.
This adaptive choice for $t$ is obtained by
making $C(\x_{n+1})$ depend only on
samples close to $\x_{n+1}$, similarly as in 
\citet{Lei2014,Barber2019,Guan2019}.

\begin{definition}[\cdsplit\ prediction band]
 \label{def:local_cde_split}
 Let $\hf(y|\x_{n+1})$
 be a conditional density estimate 
 obtained from data $\mathbb{D}'$ and
 $0 < 1-\alpha < 1$ be a coverage level.
 Let $d$ be a distance on 
 the feature space and
 $\x^{c}_1,\ldots,\x^{c}_J \in \mathcal{X}$ be 
 centroids chosen so that $d(\x^{c}_i,\x^{c}_j) > 0$.
 Consider the partition of the feature space that 
 associates each $\x \in \mathcal{X}$ to 
 the closest $\x_j^c$, i.e.,
 $\mathcal{A}=\{A_j:j=1,\ldots,J\},$ 
 where $A_j=\left\{\x \in \mathcal{X}: 
 d(\x,\x^{c}_j)< d(\x,\x^{c}_k)
 \mbox{ for every }k\neq j\right\}$.
 The \cdsplit\ prediction band for 
 $Y_{n+1}$ is:
 \begin{align*}
  C(\x_{n+1}) =
  \left\{y:\hf(y|\x_{n+1})
  \geq q(\alpha;\mathcal{T}(\x_{n+1},\mathbb{D}))\right\},
 \end{align*}
 where $\mathcal{T}(\x_{n+1},\mathbb{D})=
 \{\hf(y_i|\x_i):\x_i \in A(\x_{n+1})\}$, where 
$A(\x_{n+1})$ is the element of $\mathcal{A}$ to which $\x_{n+1}$ belongs to.
\end{definition}

\begin{remark}[Multivariate responses] 
 Although we focus on univariate targets,
 \cdsplit\ can be extended to 
 the case in which $\Y \in \mathbb{R}^p$.
 As long as an estimate of $f(\y|\x)$ is available, 
 the same construction can be applied.
\end{remark}

The bands given by \cdsplit\ control 
local coverage in the sense 
proposed by \citet{Lei2014}.

\begin{definition}[Local validity; 
Definition 1 of \citet{Lei2014}]
 Let $\mathcal{A}=\{A_j:j\geq 1\}$ be 
 a partition of $\mathcal{X}$. 
 A prediction band $C$ is locally valid 
 with respect to $\mathcal{A}$ if,
 for every $j$ and $\P$,
 \begin{align*}
  \P(Y_{n+1}\in C(\X_{n+1})|\X_{n+1} \in A_j) 
  \geq 1-\alpha
 \end{align*}
\end{definition}

\begin{theorem}[Local and marginal validity]
 \label{thm:cd_split_control}
 The \cdsplit \ band is 
 locally valid with respect to $\mathcal{A}$.
 It follows from \citet{Lei2014} that
 the \cdsplit \ band is also marginally valid.
\end{theorem}

Although \cdsplit \ controls local coverage,
its performance drastically depends on 
the chosen partition of the feature space.
If the partitions are not chosen well, 
local coverage may be far from conditional coverage.
For instance, if the partitions are
defined according to the Euclidean distance
\citep{Lei2014,Barber2019}, then
the method will not scale to
high-dimensional feature spaces.
In these settings small Euclidean neighborhoods 
have few data points and, therefore,
large neighborhoods must be taken.
As a result, local coverage is 
far from conditional coverage.
We overcome this drawback by using
a specific data-driven partition.
In order to build this metric,
we start by defining the profile of a density,
which is illustrated in \cref{fig:profile}. 

\begin{definition}[Profile of a density]
\label{def:profile}
 For every $\x \in \Re^d$ and $t \geq 0$, 
 the profile of $\hf(y|\x)$, 
 $g_{\x}(t)$, is
 \begin{align*}
  g_{\x}(t) &:= 
  \int_{\{y:\hf(y|\x)\geq t\}} \hf(y|\x)dy.
 \end{align*}
\end{definition}

\Cref{def:profile} is used to
define a distance between features,
the profile distance,
as defined below.

\begin{figure}[H]
	\centering
	\hspace{-5mm}
	\includegraphics[scale=0.35]{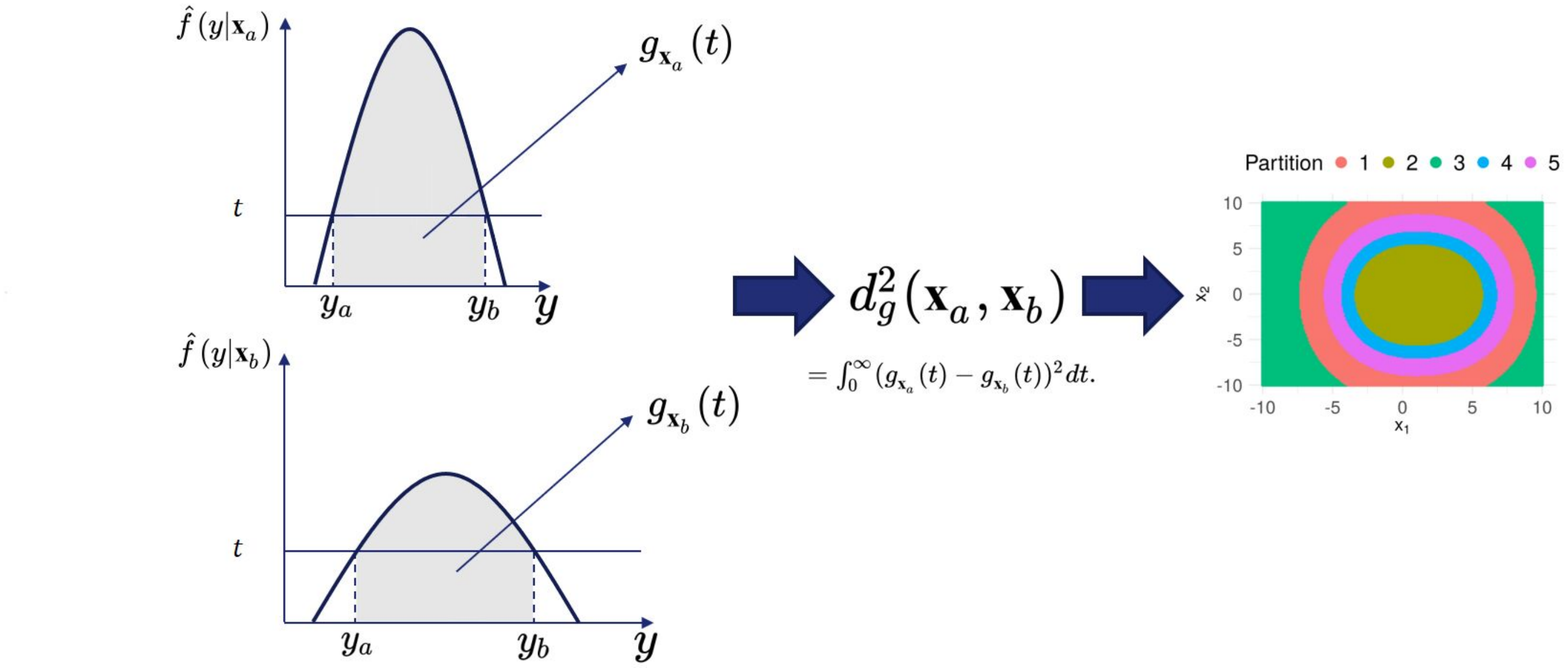}
	\caption{Illustration of the
		profile distance, which  is 
		used in \cdsplit\ for  
		partitioning the feature space.}
	\label{fig:profile}
\end{figure}

\begin{definition}[Profile distance]
 The profile distance\footnote{
  The profile distance is a metric on
  the quotient space $\mathcal{X}/\sim$, 
  where $\sim$ is the equivalence relation
  $\x_a \sim \x_b \iff 
  g_{\x^a} = g_{\x^b}$ a.e.
 }
 between  $\x_a,\x_b \in \mathcal{X}$ is
 \begin{align*} 
  d_g^2(\x_a,\x_b) &:= 
  \int_0^\infty \left(g_{\x_a}(t)-g_{\x_b}(t)\right)^2 dt,
 \end{align*}
\end{definition}

Contrary to the Euclidean distance,
the profile distance is appropriate 
even for high-dimensional data.
For instance, two points might be far
in Euclidean distance and 
still have similar conditional densities.
In this case one would like 
these points to be on the same partition element.
The profile obtains this result by
measuring the distance between instances based on
the distance between their conditional densities.
By grouping points with similar conditional densities,
the profile distance allows 
partition elements to be larger without
compromising too much the approximation of
local validity to conditional validity.
This property is illustrated in
the following examples.

\begin{example}\textbf{[Location family]}
 Let  $h(y)$ be a density, 
 $\mu(\x)$ a function, and
 $Y|\x \sim h(y-\mu(\x))$.
 In this case,
 $d_g(\x_a,\x_b) = 0$,
 for every $\x_a,\x_b \in \mathbb{R}^d$.
 For instance, if $Y|\x\sim N(\beta^t \x,\sigma^2)$, 
 then all instances have the same profile.
 Indeed, in this special scenario,
 if \cdsplit\ uses a unitary partition, then
 conditional validity is obtained.
\end{example}

\begin{example}\textbf{[Irrelevant features]}
 If $\x_S$ is a subset of the features such that
 $f(y|\x)=f(y|\x_S)$, then
 $d_g(\x_a,\x_b)$ does not depend
 on the irrelevant features, $S^c$.
 While irrelevant features do not
 affect the profile distance,
 they can have a large impact
 in the Euclidean distance in
 high-dimensional settings.
\end{example}

Also, if all samples that fall into the same partition as
$\x_{n+1}$ have the same profile 
as $\x_{n+1}$ according to $f$, then
the statistics used in \cdsplit\ ,
$\mathcal{T}(\x_{n+1},\mathbb{D})$
are i.i.d. data.
Thus, the quantile used in \cdsplit\
will be the $\alpha$ quantile of
$f(Y_{n+1}|\x_{n+1})$. 
This in turn makes
$C(\x_{n+1})$ the smallest prediction band
with conditional validity of $1-\alpha$.
\Cref{thm:convergeHPD}, below,
formalizes this statement.
\begin{theorem}
 \label{thm:convergeHPD}
 Assume that all samples that fall into the same partition as $\x_{n+1}$, say $(\X_1,Y_1),\ldots,(\X_m,Y_m)$, are such
that $g_{\x_i}=g_{\x_{n+1}}$, and that $\hf(y|\x)=f(y|\x)$ is continuous as a function of $y$ for every $\x \in \mathcal{X}$.
Let $T_m:=q(\alpha;\mathcal{T}(\x_{n+1},\mathbb{D}))$ be the cutoff
used in \cdsplit. Then, for every fixed $\alpha \in (0,1)$
$$T_m \xrightarrow[a.s.]{m\longrightarrow \infty} t^*$$
where $t^*=t^*(\x_{n+1},\alpha)$ is the cutoff associated to the oracle band (i.e., the smallest predictive region with coverage $1-\alpha$).
\end{theorem}

Given the above reasons,
the profile density captures 
what is needed of a meaningful neighborhood that 
contains many samples even in high dimensions.  
Indeed, consider a partition
of the feature space, $\mathcal{A}$,
that  has the property that all samples that 
belong to the same element of $\mathcal{A}$ 
have the same  oracle cutoff $t^*$. 
\Cref{thm:equivalance} shows that
the coarsest
partition
that 
has this property
is the one  induced by the profile distance.

\begin{theorem}
\label{thm:equivalance}
Assume that $\hf(y|\x)=f(y|\x)$
is continuous as a function of $y$ for every $\x \in \mathcal{X}$. For each sample $\x \in \mathcal{X}$ and miscoverage level $\alpha \in (0,1)$, let $t^*(\x,\alpha)$ be the cutoff of the oracle band
for $f(y|\x)$ with coverage $1-\alpha$.
Consider the equivalence relation $\x_a \sim \x_b \iff d_g(\x_a,\x_b)=0$. Then 
\begin{enumerate}[label=(\roman*)]
    \item if\
$\x_a \sim \x_b$, then 
$t^*(\x_a,\alpha)=t^*(\x_b,\alpha)$ for every $\alpha \in (0,1)$
\item if $\sim'$ is any other equivalence relation such that
$\x_a \sim' \x_b$ implies that
$t^*(\x_a,\alpha)=t^*(\x_b,\alpha)$ for every $\alpha \in (0,1)$, then 
$\x_a \sim' \x_b \Rightarrow \x_a \sim \x_b$.
\end{enumerate}
\end{theorem}

We therefore always use \cdsplit\ with the profile distance.
In order to compute the prediction bands, we still need to define
the centroids $\x_i^c$.
Ideally, the partitions should be such that (i) all sample points inside a given element of the partition have similar profile, and (ii) sample points that belong to different elements of the partition have profiles that are very different from each other.
We accomplish this by choosing the partitions by applying a k-means++ clustering algorithm
\citep{arthur2007k}, but
using the profile distance instead of the Euclidean one.
This is done by
applying the standard (Euclidean)
k-means++ algorithm
to the data
points 
$\w_i:=\tilde{g}(\x_i)$, where 
$\tilde{g}(\x_i)$ is a discretization of the function
$g(\x_i)$, obtained by evaluating
$g(\x_i)$
on a grid of values.
$\w_1^c$,\ldots,$\w_J^c$  are then the centroids of such clusters. Figure \ref{fig:partitions} illustrates the partitions that are obtained in one dataset.
The profile distance allows samples that are far from each other in the Euclidean sense to fall into the same element of the partition. This is the key reason why our method scales to high-dimensional datasets.

\begin{figure}
 \centering
 \hspace{-10mm}\includegraphics[scale=0.32]{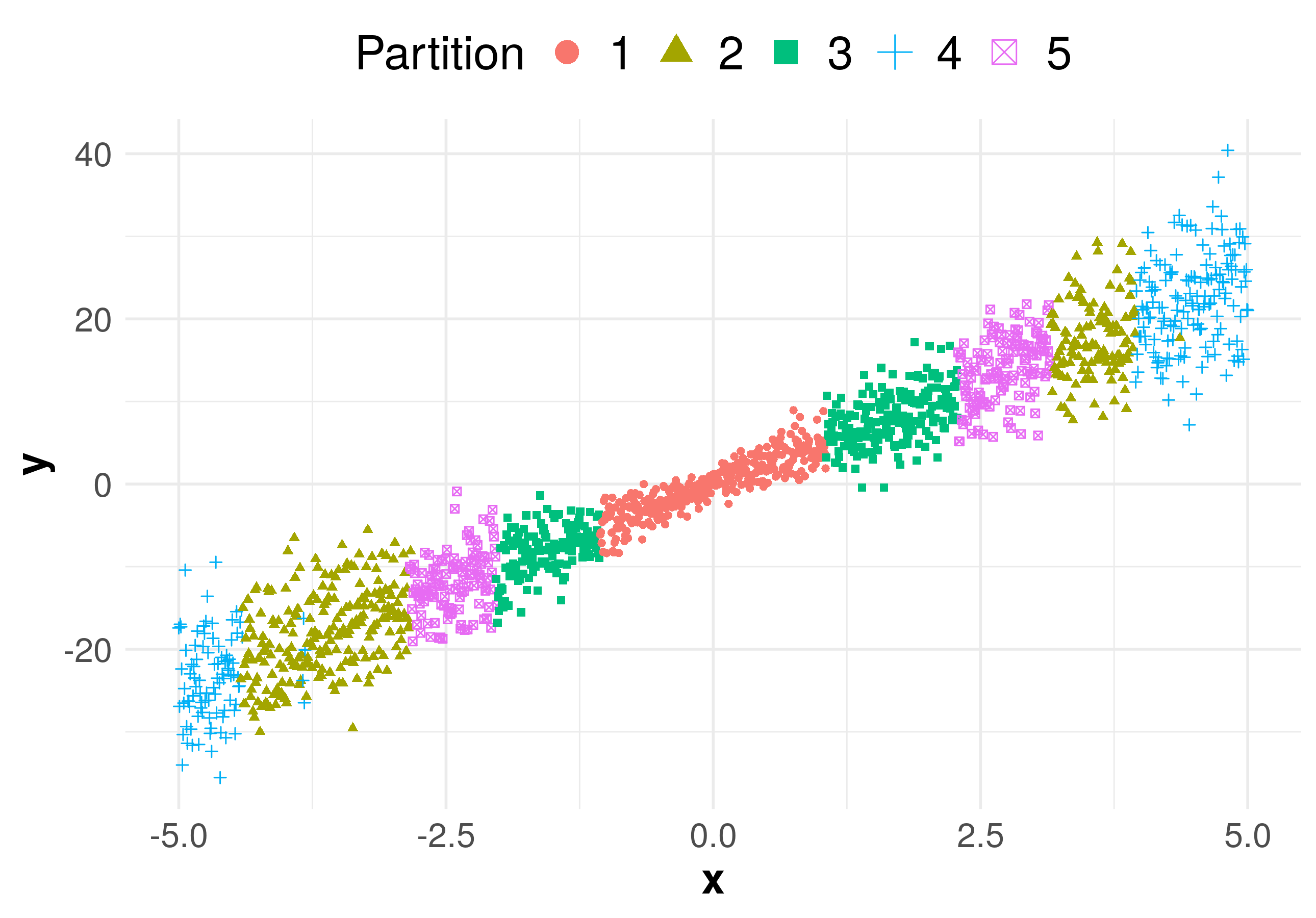}
 \caption{Scatter plot of data generated according to $Y|x \sim N(5x,1+|x|)$. Colors indicate partitions that were obtained using the profile of the estimated densities. Because  points that are far from each other on the $x$-axis can belong to the same element of the partition, this approach to  partition
 the feature space increases the number of data points that belong to each element. These elements however preserve the optimal cutoff (Theorem \ref{thm:equivalance}).}
 \label{fig:partitions}
\end{figure}

Algorithm \ref{alg:cd} shows pseudo-code for implementing  \cdsplit.

\begin{algorithm}
  \caption{ \small \cdsplit}\label{alg:cd}
  \algorithmicrequire \ {\small Data $(\x_{i},Y_{i})$, $i=1,...,n$, miscoverage level $\alpha \in (0,1)$, algorithm $\mathcal{B}$ for fitting  conditional density function, number of elements of the partition $J$.}\\ 
   \algorithmicensure \ {\small Prediction band for $\x_{n+1}\in\mathbb{R}^d$}
    \begin{algorithmic}[1]
    \STATE Randomly split $\{1,2,...,n\}$ into two subsets $\mathbb{D}$ and $\mathbb{D}'$
     \STATE Fit $\hat{f} = \mathcal{B}(\{(\x_i,Y_i):i \in \mathbb{D}' \})$ \codecomment{Estimate cumulative density function}
     \STATE Compute $\mathcal{A}$, the partition of $\mathcal{X}$, by applying 
     k-means++ on the profiles of the samples in $\mathbb{D}$'
          \STATE 
     Compute $g_{\x_{n+1}}(t)=\int_{\{y:\hat{f}(y|\x)\geq t\}}  \hat{f}(y|\x)dy$, for all $t>0$ \codecomment{Profile of the density (Definition \ref{def:profile})}
     \STATE 
     Find $A(\x_{n+1}) \in \mathcal{A}$, the element
     of $\mathcal{A}$
     such that $\x_{n+1} \in \mathcal{A}$

     \STATE 
     Compute $g_{\x_i}(t)=\int_{\{y:\hat{f}(y|\x)\geq t\}}  \hat{f}(y|\x)dy$, for all $t>0$ and $i \in \mathbb{D}$ \codecomment{Profile of the densities (Definition \ref{def:profile})}
     \STATE  Let $\mathcal{T}(\x_{n+1},\mathbb{D})=
    \{\hat{f}(y_i|\x_i), i \in \mathbb{D}: \x_i \in A(\x_{n+1})\}$
     \STATE Let $t=q(\alpha;\mathcal{T}(\x_{n+1},\mathbb{D}))$ \codecomment{Compute the $\alpha$- quantile of the set $\mathcal{T}(\x_{n+1},\mathbb{D})$}
     \STATE \textbf{return} $\left\{y:\hat{f}(y|\x^{*})\geq t\right\}$
  \end{algorithmic}
  \end{algorithm}



\subsection{Multiclass classification}

If the sample space $\mathcal{Y}$ is discrete,
 we use a similar construction to that of Definition \ref{def:local_cde_split}.
 More precisely, the \cdsplit\ prediction band is given by
 \begin{align*}
    C(\x_{n+1})=\left\{y:\widehat{\P}(Y=y|\x_{n+1})\geq q(\alpha;\mathcal{T}(\x_{n+1},\mathbb{D}))\right\},
\end{align*}
where 
\begin{align*}
    \mathcal{T}(\x_{n+1},\mathbb{D})=
    \left\{\widehat{\P}(Y_i=y_i|\x_i), i = 1,\ldots,n: \x_i \in A(\x_{n+1})\right\},
\end{align*}
$A(\x_{n+1})$ is the element of $\mathcal{A}$ to which $\x_{n+1}$ belongs to, and
$$d_g^2(\x_a,\x_b)=\sum_{y \in \mathcal{Y}} \left(\widehat{\P}(Y=y|\x_a)-\widehat{\P}(Y=y|\x_b)\right)^2.$$

Theorems analogous to those presented in the last section hold in the  classification
setting as well.
\begin{remark}
While  \cdsplit\ is developed to control the coverage of $C$ conditional on the value $\x_{n+1}$, in a classification setting some methods  (e.g.  \citealt{Sadinle2019}) control \emph{class-specific coverage}, defined as  $$\P(Y_{n+1} \in C(\X_{n+1})|Y_{n+1}=y)\geq 1-\alpha_y.$$ 
\end{remark}

\section{Experiments}
\label{sec:exp}

We consider the following settings
with $d=20$ covariates:
\begin{itemize}
    \item\textbf{[Asymmetric]} 
    $\X=(X_1,\ldots,X_d)$, with $X_i \overset{\text{iid}}{\sim} \mbox{Unif}(-5,5)$, and $Y|\x = 5x_1 + \epsilon$,
    where $\epsilon \sim \mbox{Gamma}(1+2|x_1|,1+2|x_1|)$
    \item\textbf{[Bimodal]}
    $\X=(X_1,\ldots,X_d)$, with $X_i \overset{\text{iid}}{\sim} \mbox{Unif}(-1.5,1.5)$, and $Y|\x \sim 0.5\mbox{N}(f(\x)-g(\x),\sigma^2(\x))+0.5\mbox{N}(f(\x)+g(\x),\sigma^2(\x)),$
    with $f(\x)=(x_1-1)^2 (x_1+1),$ $g(\x)=2\I(x_1 \geq -0.5) \sqrt{x_1+0.5}$,
    and $\sigma^2(\x)=1/4+|x_1|$.
    This is the example from \citep{Lei2014} with $d-1$ irrelevant variables added.
    \item\textbf{[Heteroscedastic]}
    $\X=(X_1,\ldots,X_d)$, with $X_i \overset{\text{iid}}{\sim} \mbox{Unif}(-5,5)$, and $Y|\x \sim \mbox{N}(x_1,1+|x_1|)$
    \item\textbf{[Homoscedastic]}
    $\X=(X_1,\ldots,X_d)$, with $X_i \overset{\text{iid}}{\sim} \mbox{Unif}(-5,5)$, and $Y|\x \sim \mbox{N}(x_1,1)$
\end{itemize}

We compare the performance of the following methods:
\begin{itemize}
 \item \textbf{[Reg-split]} 
 The regression-split method \citep{Lei2018}, 
 based on the conformal score 
 $|Y_i-\widehat{r}(\x_i)|$, where 
 $\widehat{r}$ is an estimate of
 the regression function.
 
 \item \textbf{[Local Reg-split]} 
 The local regression-split method \citep{Lei2018}, 
 based on the conformal score $\frac{|Y_i-\widehat{r}(\x_i)|}{\widehat{\rho}(\x_i)}$,
 where $\widehat{\rho}(\x_i)$ is
 an estimate of the conditional mean
 absolute deviation of $|Y_i-r(\x_i)|\x_i$.

 \item \textbf{[Quantile-split]} 
 The conformal quantile regression method \citep{Romano2019,Sesia2019}, based on
 conformalized quantile regression. 

 \item \textbf{[Dist-split]} 
 From \cref{sec:dist}.

 \item \textbf{[CD-split]} 
 From \cref{sec:cd} with
 partitions of size $J=\ceil{n/100}$.
\end{itemize}

For each sample size, $n$, 
we use a coverage level of $1-\alpha=90\%$ and
run each setting 5,000 times.
In order to make fair comparisons between 
the various approaches,
we use random forests \citep{Breiman2001} to
estimate all quantities needed, namely:
the regression function in Reg-split,
the conditional mean absolute deviation 
in Local Reg-split, 
the conditional quantiles via 
quantile forests \citep{Meinshausen2006} 
in Quantile-split,
and the conditional density 
via FlexCode \citep{Izbicki2017}
in \distsplit\ and \cdsplit .
A conditional cumulative distribution estimate,
$\widehat{F}(y|\x)$ is obtained by 
integrating the conditional density estimate:
$\widehat{F}(y|\x) = 
\int_{-\infty}^y \widehat{f}(y|\x)dy$.
The tuning parameters of all methods were set to be the default values of 
 the packages that were used.

\begin{figure*}[hp!]
        \centering
\subfloat{ \includegraphics[scale=0.31]{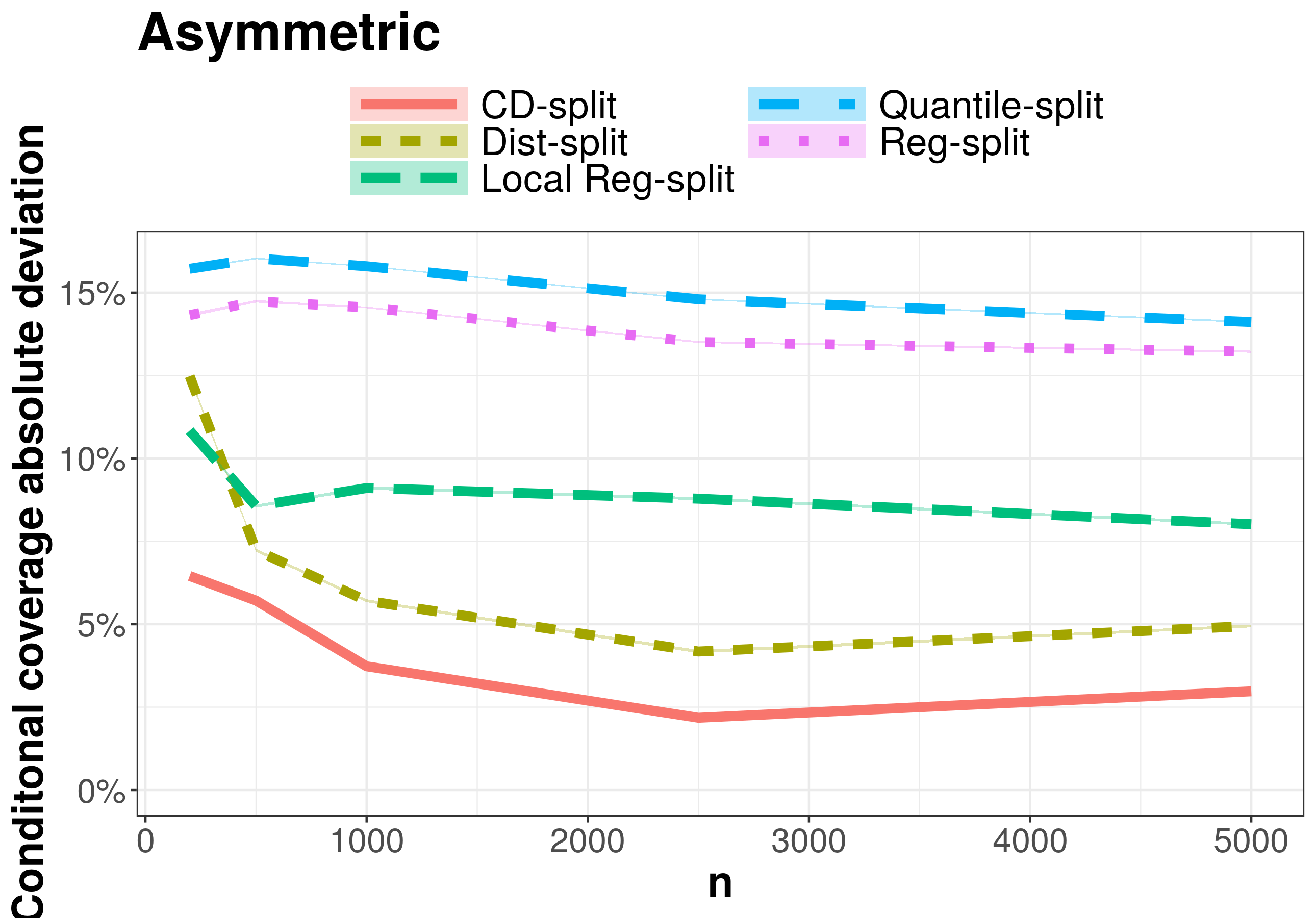}}\hspace{2mm}
\subfloat{ \includegraphics[scale=0.31]{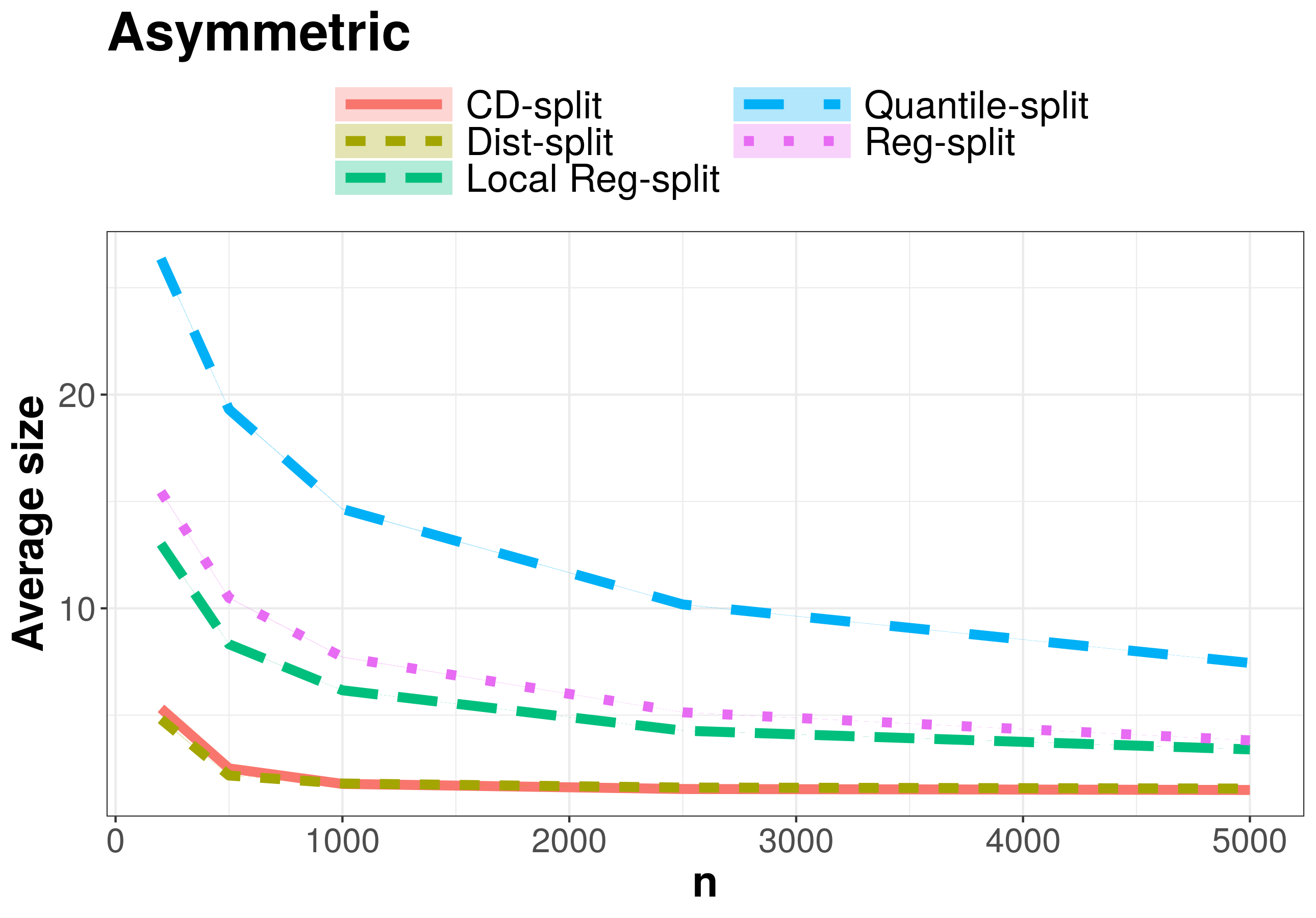}}\\[2mm]
\subfloat{\includegraphics[scale=0.32]{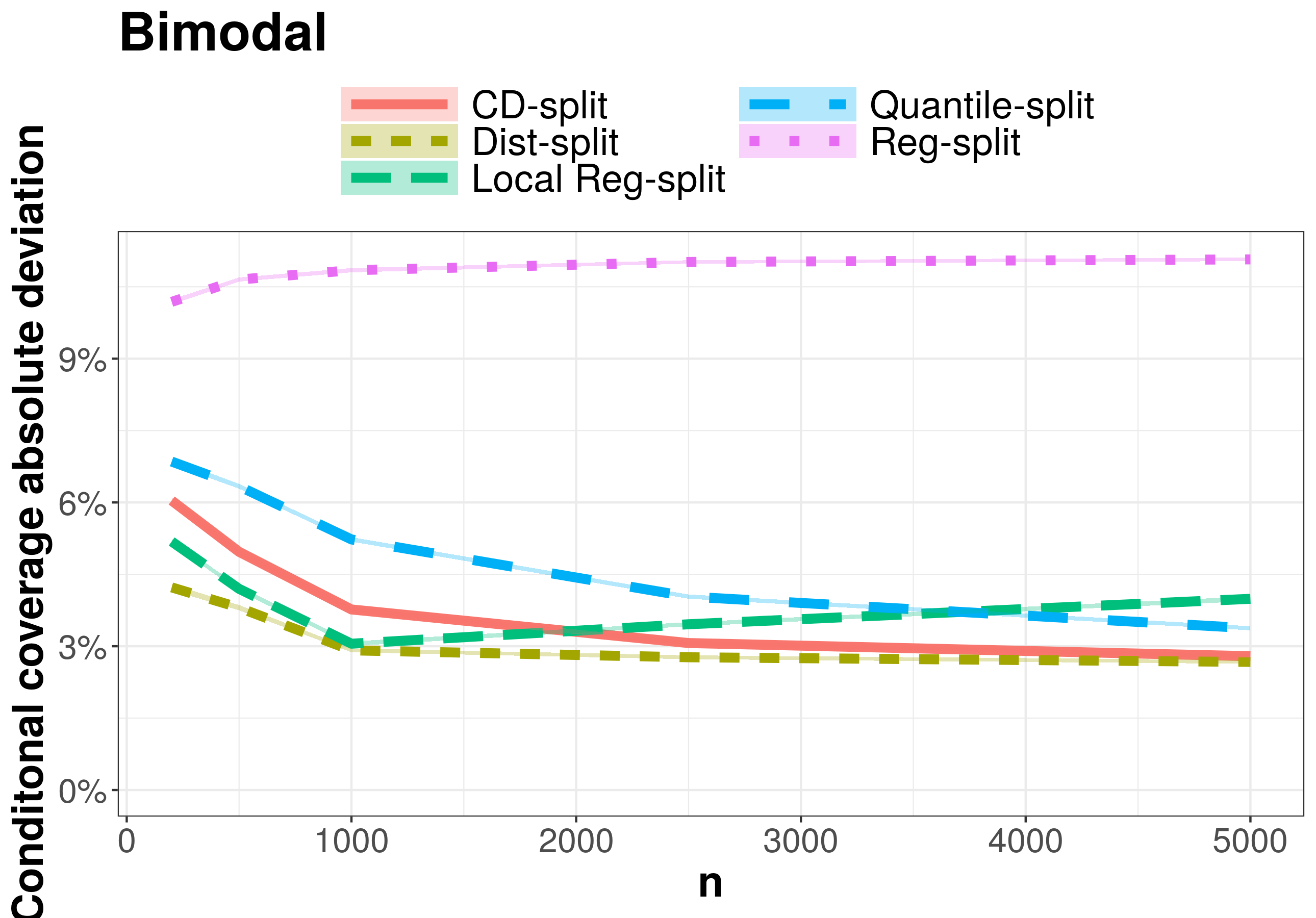}}\hspace{2mm} 
\subfloat{ \includegraphics[scale=0.31]{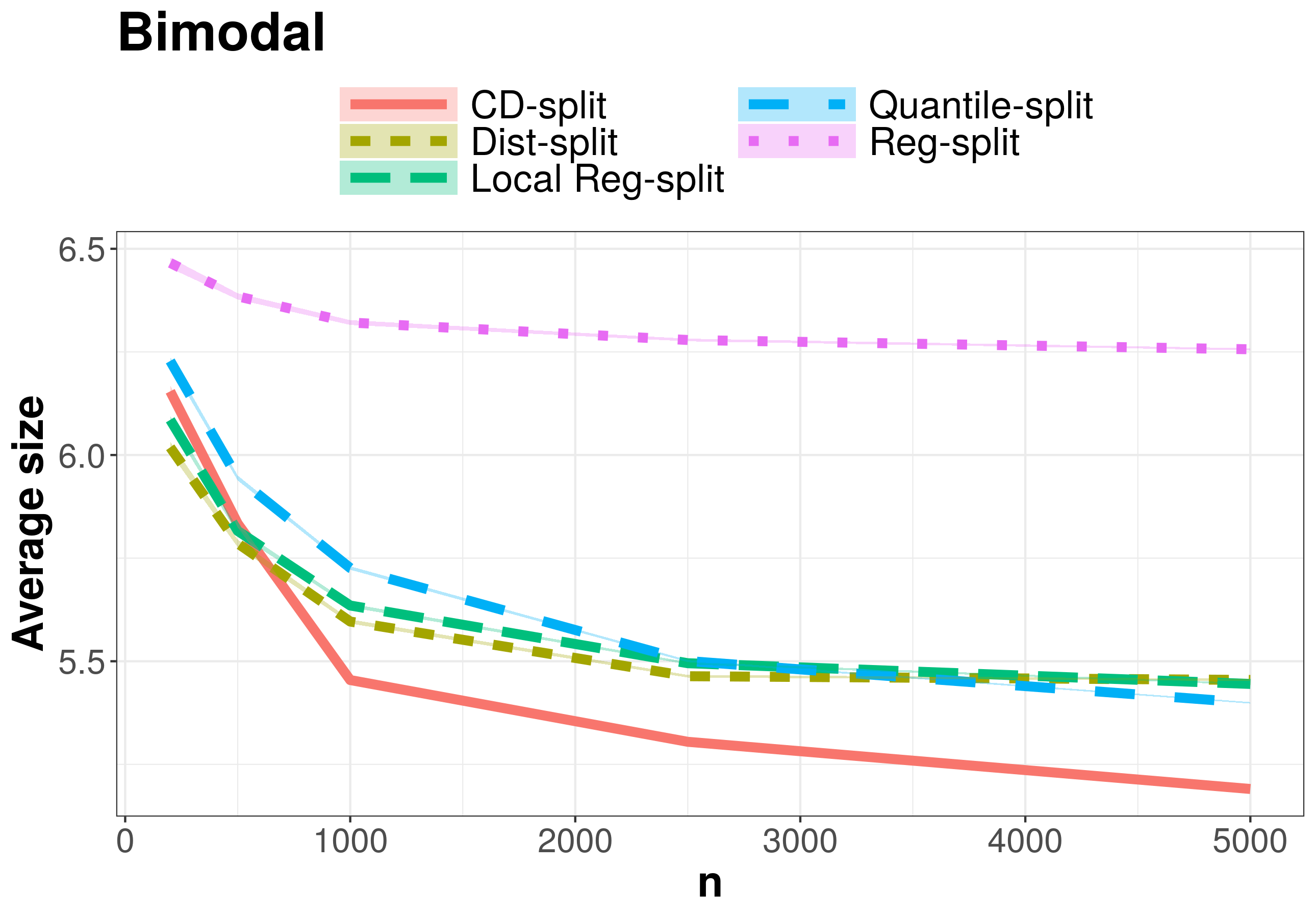}}\\[2mm]
\subfloat{\includegraphics[scale=0.32]{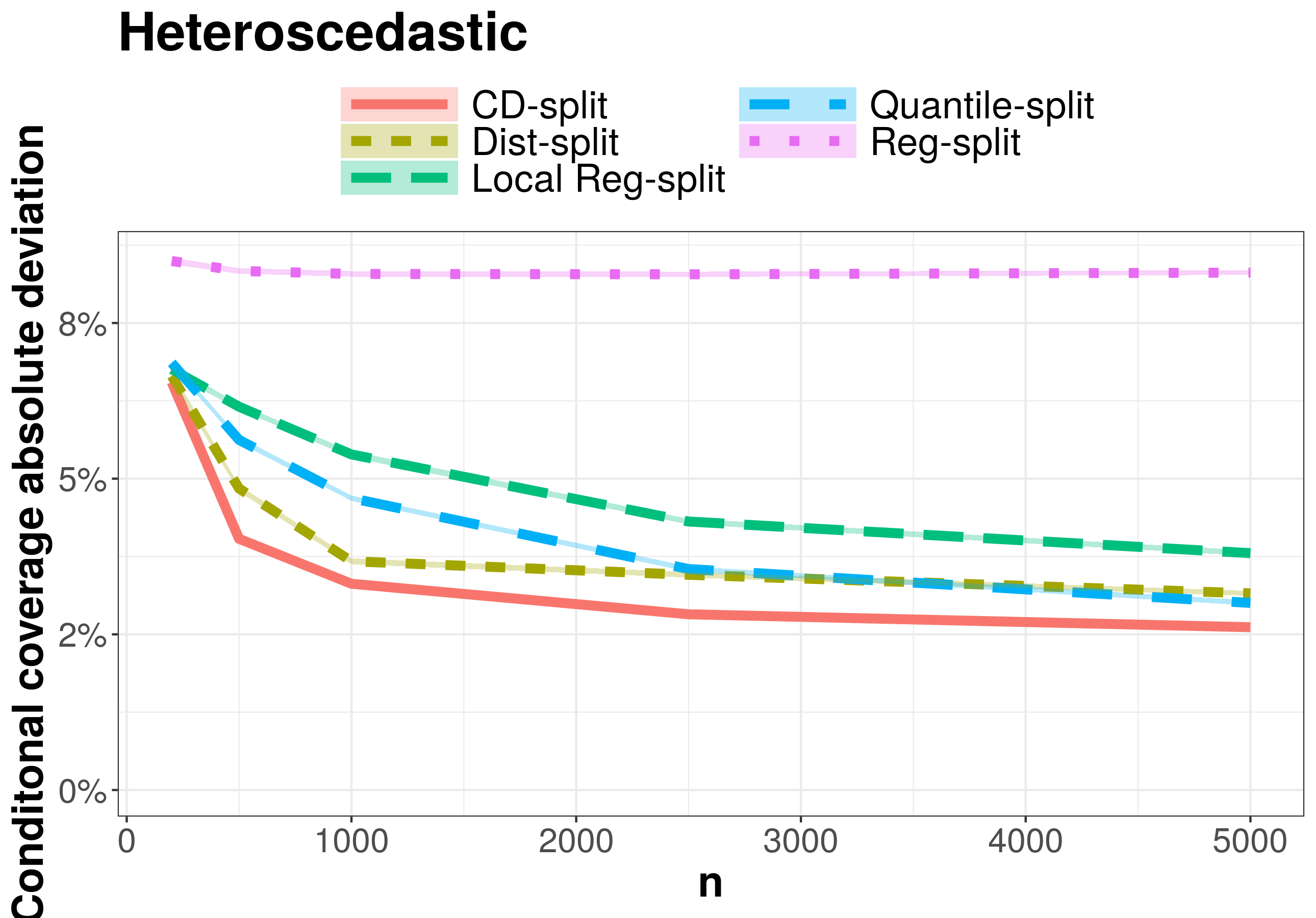}}\hspace{2mm} 
\subfloat{ \includegraphics[scale=0.31]{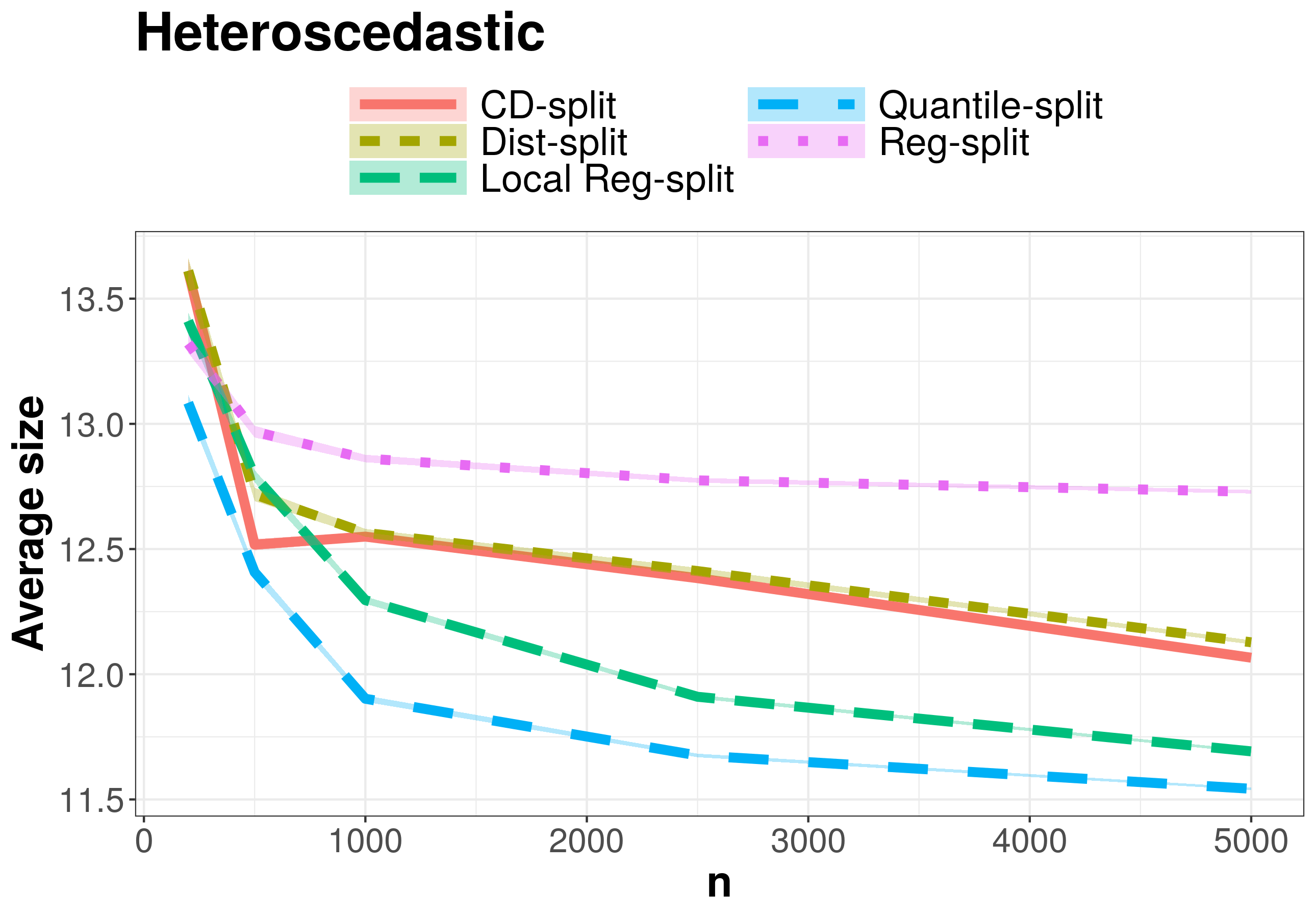}}\\[2mm]
\subfloat{\includegraphics[scale=0.31]{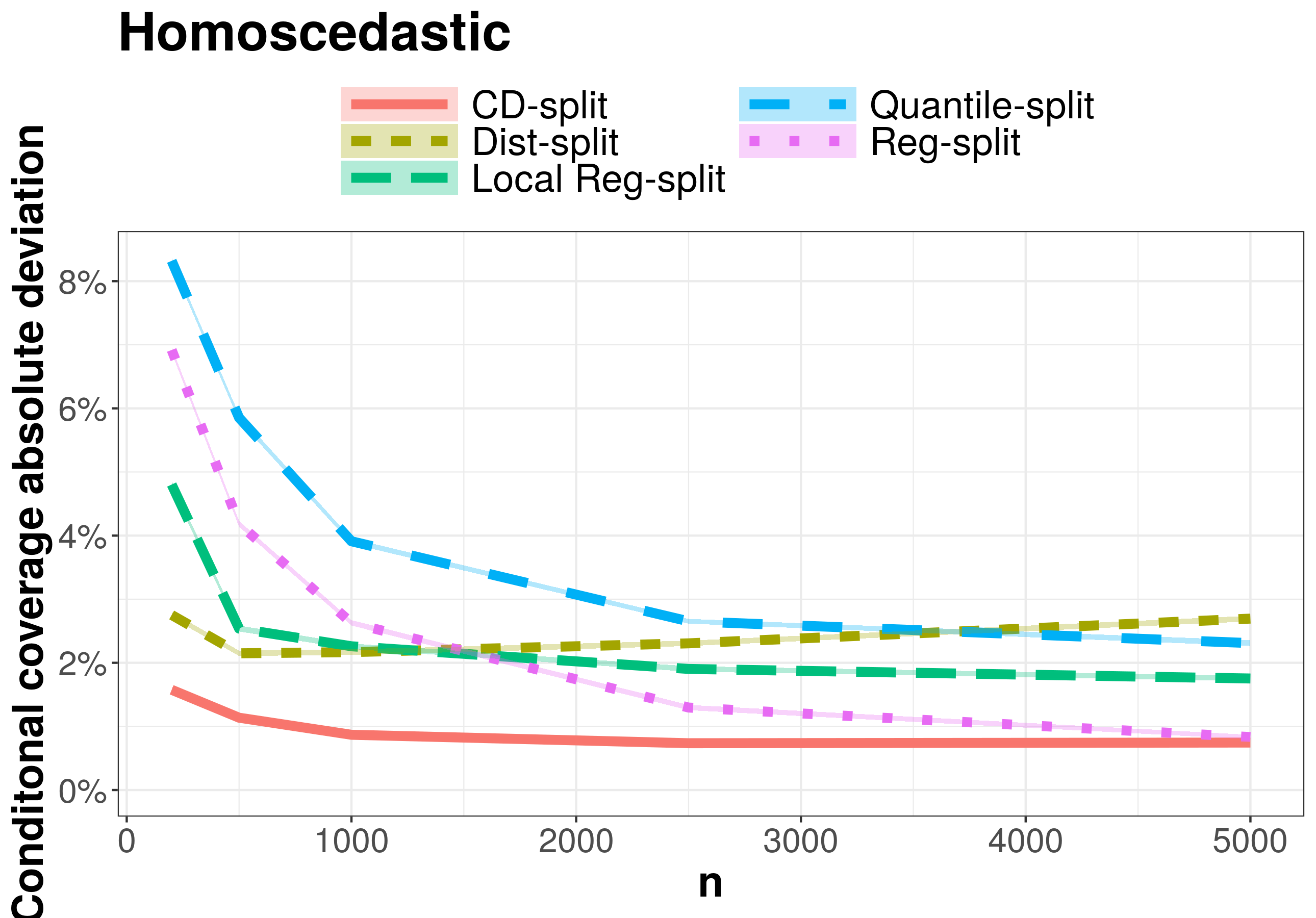}}\hspace{2mm} 
\subfloat{ \includegraphics[scale=0.31]{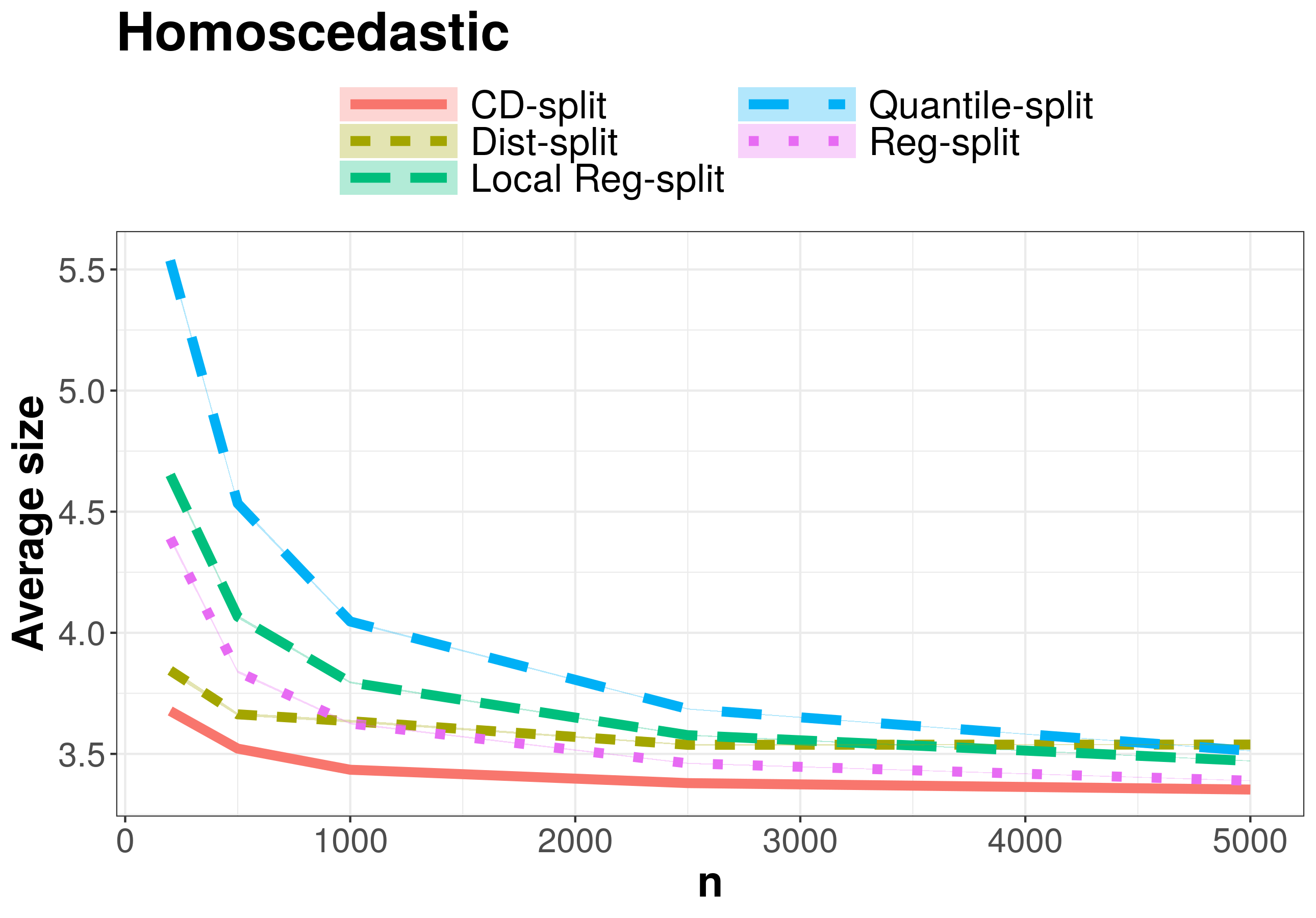}}
\vspace{-2mm} 
        \caption{\footnotesize
        Performance of each regression method as a function of the sample size. Left panels show how much the conditional coverage vary with $\x$;  right panels display the average size of the prediction bands.} 
\label{fig:regression} 
\end{figure*}

\Cref{fig:regression} shows the performance of 
each method as a function of the sample size.
While the left side figures display how well
each method controls conditional coverage,
the right side displays the average size
of the regions that are obtained.
The control of the conditional coverage is
measured through the conditional coverage
absolute deviation, that is,
$\E[|\P(Y^* \in C(\X^*)|\X^*)-(1-\alpha)|]$.
Since all of the methods obtain marginal coverage
very close to the nominal $90\%$ level,
this information is not displayed in the figure.
\Cref{fig:regression} shows that,
in all settings, \cdsplit\ is
the method which best controls conditional coverage.
Also, in most cases its 
prediction bands also have the smallest size.
Similarly, \distsplit\ frequently is 
the second method with
both highest control of conditional coverage 
and also smallest prediction bands.

We also apply \cdsplit \ to 
a classification setting. 
We consider  $\X=(X_1,\ldots,X_d)$,
with $X_i \overset{\text{iid}}{\sim} N(0,1)$
and $Y|\X$ follows the logistic model,
$\P(Y=i|\x) \propto 
\exp\left\{\boldsymbol{\beta} \cdot \x \right\}$,
where $\boldsymbol{\beta}=(-6, -5, -1.5, 0, 1.5, 5, 6)$. 
We compare \cdsplit\ to Probability-split, 
the method described in \citet[Sec. 4.3]{Sadinle2019}, 
which has the goal of controlling global coverage. Probability-split is a particular case of \cdsplit: 
it corresponds to applying \cdsplit \ 
with $J=1$ partitions.
\Cref{fig:classification} shows the results. 
\cdsplit\ better controls conditional coverage. 
On the other hand, its prediction bands are,
on average, larger than those of Probability-split.
 
\begin{figure*}[h!]
        \centering
\subfloat{ \includegraphics[scale=0.32]{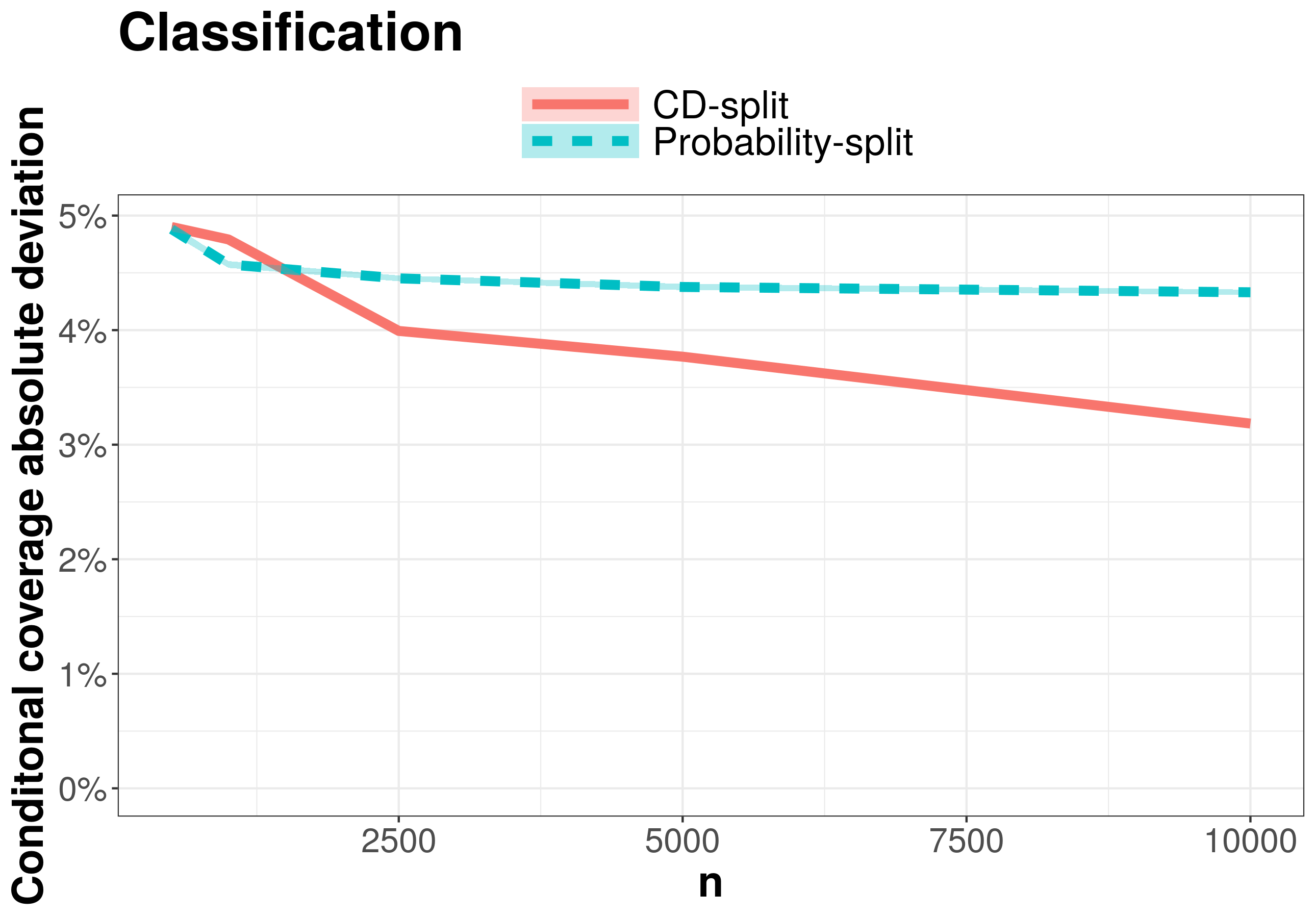}}\hspace{2mm}
\subfloat{ \includegraphics[scale=0.32]{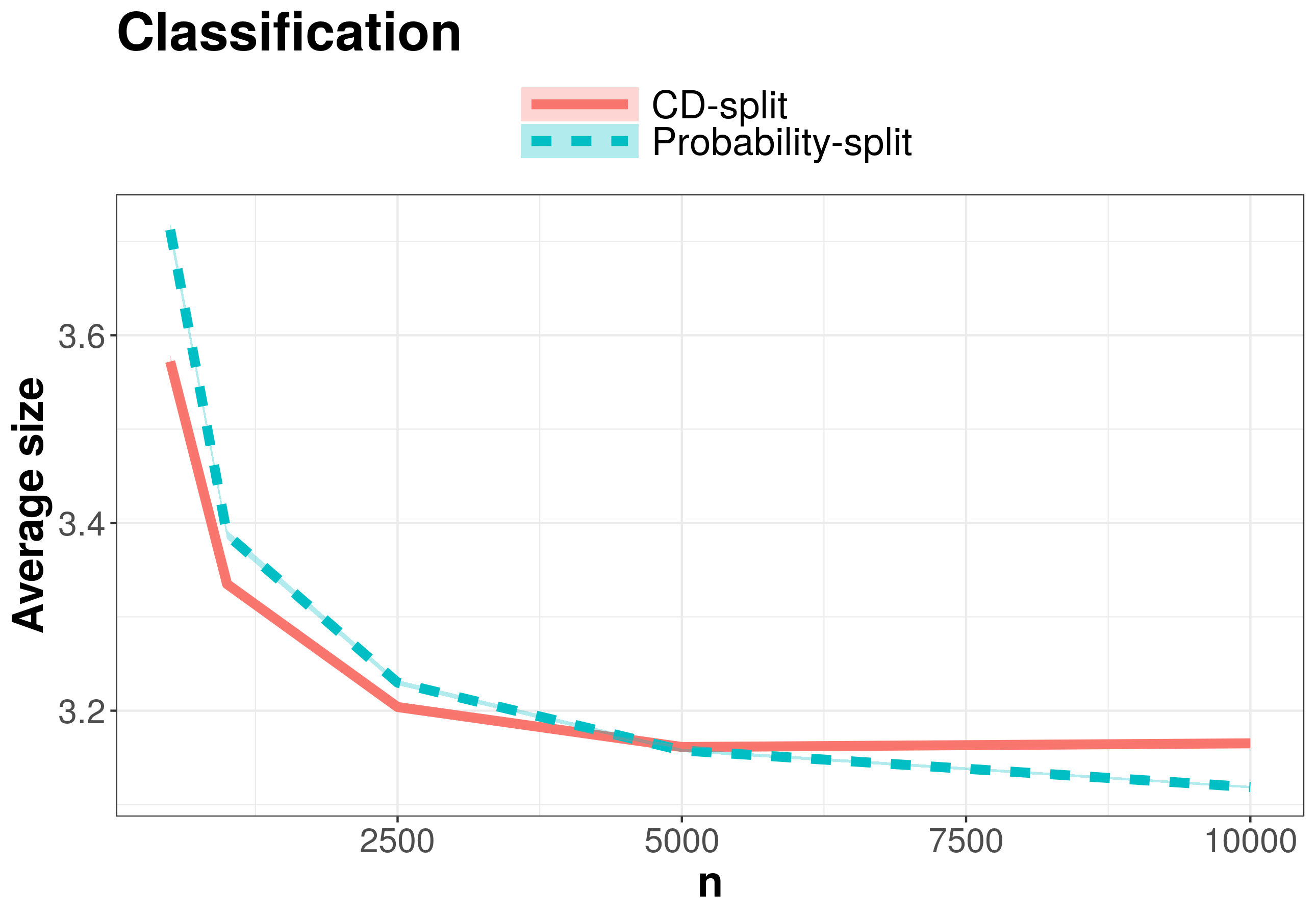}}
\vspace{-2mm}
        \caption{\footnotesize
        Performance of each classification method as a function of the sample size. Left panel shows how much the conditional coverage vary with $\x$;  right panel displays the average size of the prediction bands.} 
\label{fig:classification} 
\end{figure*}

\section{Final remarks}
\label{sec:final}

We introduce \distsplit\ and \cdsplit ,
which obtain asymptotic conditional coverage
and converge to optimal oracle bands,
even in high-dimensional feature spaces.
These results do not require assumptions
about the dependence between
the target variable and the features.
Both methods are based on estimating conditional densities.
While  \distsplit\ necessarily leads to intervals, 
which are easier to interpret, 
\cdsplit\ leads to smaller prediction regions.
A simulation study shows that
both methods yield smaller prediction bands and
better control of conditional coverage than
other methods in the literature
under a variety of settings.
We also show that \cdsplit \ leads to
good results in classification problems.

\cdsplit\ is based on
a novel data-driven metric on the feature space that
appropriate for defining neighborhoods for conformal methods,
in particular in high-dimensional settings.
It might be possible to use this metric with
other conformal methods to obtain
asymptotic conditional coverage.

R code for implementing \distsplit \ and \cdsplit \ is
available on \url{https://github.com/rizbicki/predictionBands}.

\subsubsection*{Acknowledgements}

This study was financed in part by the Coordena\c{c}\~{a}o de Aperfei\c{c}oamento de Pessoal de N\'{\i}vel Superior - Brasil (CAPES) - Finance Code 001. 
Rafael Izbicki is grateful for the financial support of FAPESP (grants 2017/03363-8 and 2019/11321-9) and CNPq (grant 306943/2017-4).
The authors are also grateful for the suggestions given by Luís Gustavo Esteves and Jing Lei.

\bibliography{main}

\section*{Proofs}

\begin{definition}
 $U_{\lfloor\alpha \rfloor}$ and 
 $U_{\lceil\alpha \rceil}$ are the
 $\lfloor n^{-1}(n\alpha)\rfloor$
 $\lceil n^{-1}(n\alpha)\rceil$
 empirical quantiles of
 $U_1,\ldots, U_n$,
\end{definition}

\subsection*{Related to \distsplit}

\begin{proof}[Proof of Theorem \ref{thm:dist_split_control}]
 Let $U_i = \hF(Y_i|\X_i)$.
 Since $(\X_i,Y_i)$ are i.i.d.
 continuous random variables and
 $\hF$ is continuous, obtain that
 $U_i$ are i.i.d. continuous random variables.
 \begin{align*}
  1-\alpha &\leq \P\left(U_{n+1} \in
  [\ufloor; \uceil] \right)
  \leq 1 - \alpha + (n+1)^{-1}.
 \end{align*}
 The conclusion follows from noticing that
 \begin{align*}
  &\P\left(U_{n+1} \in
  [\ufloor; \uceil] \right) \\
  = & \P\left(Y_{n+1} \in
  [\hF^{-1}(\ufloor|\X_{n+1}); \hF^{-1}(\uceil|\X_{n+1})] \right) \\
  = & \P(Y_{n+1}\in C(\X_{n+1}))
 \end{align*}
\end{proof}

\begin{lemma}
 \label{lemma:dist-1}
 Let $I_1 = \left\{i \leq n: 
 |\hF(Y_i|\X_i)-F(Y_i|\X_i)| < \eta_n^{1/3}\right\}$
 and $I_2 = \{1,\ldots,n\}-I_1$.
 Under \cref{ass:consistent_cde},
 $|I_2| = o_P(n)$ and
 $|I_1| = n+o_P(n)$.
\end{lemma}

\begin{proof}
 Let $A_n = \left\{\E\left[\sup_{y \in \sY}
 \left(\hF(y|\X)-F(y|\X)\right)^2 \big |\hF\right] 
 \geq \eta_n\right\}$ and
 $B_n = \left\{|\hF(Y|\X)-F(Y|\X)| 
 \geq \eta_n^{1/3}\right\}$.
 \begin{align*}
  \P(B_n) 
  &= \E[\P(B_n|\hF) \I(A_n)] 
  + \E[\P(B_n|\hF) \I(A_n^c)] \\
  &\leq \P(A_n) + \E\left[\frac{\E[(\hF(Y|\X)-F(Y|\X))^2|\hF]}
  {\eta_n^{2/3}}\I(A_n^c)\right] \\
  &\leq \rho_n + \eta_n^{1/3} = o(1)
 \end{align*}
 Note that $|I_2| \sim \text{Binomial}(n, \P(B_n))$.
 Since $\P(B_n) = o(1)$, conclude that
 $|I_2| = o_P(n)$. That is,
 $|I_1| = n+o_P(n)$.
\end{proof}

\begin{lemma}
 \label{lemma:dist-2}
 Under \cref{ass:consistent_cde},
 If $U_i = \hF(Y_i|\X_i)$, then
 for every $\alpha \in (0,1)$,
 $U_{\lfloor \alpha \rfloor} = \alpha + o_{P}(1)
 = U_{\lceil \alpha \rceil}$.
\end{lemma}

\begin{proof}
 Let $I_1$ and $I_2$ be 
 such as in \cref{lemma:dist-1}.
 Also, let $\hat{G}_1$, $G_1$ and $G_0$ be,
 the empirical quantiles of, respectively,
 $\{U_i: i \in I_1\}$,
 $\{F(Y_i|\X_i): i \in I_1\}$, and
 $\{F(Y_i|\X_i): i \leq n\}$.
 By definition of $I_1$,
 for every $\alpha^* \in [0,1]$,
 $\hat{G}_1^{-1}(\alpha^*) = G_1^{-1}(\alpha^*) + o(1)$.
 Also, $G_0^{-1}(\alpha^*) = \alpha^* + o_P(1)$.
 Therefore, since
 \begin{align*}
  G_0^{-1}\left(\frac{|I_1|\alpha^*}{n}\right)
  \leq G_1^{-1}(\alpha^*)
  \leq G_0^{-1}\left(\frac{|I_1|\alpha^*+|I_2|}{n}\right),
 \end{align*}
 conclude that $\hat{G}_1^{-1}(\alpha^*) 
 = \alpha^* + o_P(1)$.
 Finally, since
 \begin{align*}
  \hat{G}^{-1}\left(\frac{n\alpha-|I_2|}{|I_1|}\right)
  \leq U_{\lfloor \alpha \rfloor} \leq 
  U_{\lceil \alpha \rceil} \leq
  \hat{G}^{-1}\left(\frac{n\alpha}{|I_1|}\right),
 \end{align*}
 Conclude that $U_{\lfloor \alpha \rfloor} = \alpha + o_P(1)
 = U_{\lceil \alpha \rceil}$.
\end{proof}

\begin{lemma}
 \label{lemma:dist-cond}
 Let $U_i = \hat{F}(Y_i|\X_i)$.
 Under \cref{ass:consistent_cde,ass:continuity},
 \begin{align*}
 \hF^{-1}(U_{[0.5\alpha]}|\X_{n+1})
 &= F^{-1}(0.5\alpha|\X_{n+1}) + o_{P}(1) \\
 \hF^{-1}(U_{[1-0.5\alpha]}|\X_{n+1})
 &= F^{-1}(1-0.5\alpha|\X_{n+1}) + o_{P}(1)
 \end{align*}
\end{lemma}

\begin{proof}
 In order to prove the first equality,
 it is enough to show that
 $F^{-1}(U_{[0.5\alpha]}|\X_{n+1})
 = F^{-1}(0.5\alpha|\X_{n+1}) + o_P(1)$
 and that $\hF^{-1}(U_{[0.5\alpha]}|\X_{n+1}) 
 = F^{-1}(U_{[0.5\alpha]}|\X_{n+1}) + o_{P}(1)$.
 The first part follows from
 \cref{lemma:dist-2} and
 the continuity of $F(y|\x)$
 (\cref{ass:continuity}).
 For the second part, note that,
 if $\sup_y |\hat{F}(y|\x)-F(y|\x)| < \eta_n$,
 then, for every $\alpha^*$,
 $|\hat{F}^{-1}(\alpha^*)-F^{-1}(\alpha^*)| 
 \leq \eta_n 
 \left(\inf_y \frac{dF(y|\x)}{dy}\right)^{-1}$.
 Using this observation,
 the proof of the second part follows from
 \cref{ass:continuity}, and observing that
 $U_{[0.5\alpha]} = 0.5\alpha + o_{P}(1)$
 (\cref{lemma:dist-2}) and
 $\P(\sup_y |\hat{F}(y|\x)-F(y|\x)| 
 \geq \eta_n) = o(1)$
 (\cref{ass:consistent_cde}).
 
 The proof for the $1-.5\alpha$ quantile
 is analogous to the one for
 the $.5\alpha$ quantile.
\end{proof}

\begin{proof}[Proof of \cref{thm:dist_optimal}]
 Follows directly from
 \cref{lemma:dist-cond}.
\end{proof}

\subsection*{Related to \cdsplit}

\begin{proof}[Proof \cref{thm:cd_split_control}]
 Let$\{i_1,\ldots,i_{n_j}\} = \{i: \X_i \in A(\x_{n+1})\}$,
 $U_l = \hf(Y_{i_l}|\X_{i_l})$, 
 for $l=1,\ldots,n_j$, and 
 $U_{n_j+1} = \hf(Y_{n+1}|\X_{n+1})$. 
 Since $(\X_1,Y_1),\ldots,(\X_{n_j},Y_{n_j}),
 (\X_{n+1},Y_{n+1})$ are i.i.d. 
 random variables, obtain that 
 $U_i$ are i.i.d. random variables 
 conditional on the event 
 $\X_{n+1} \in A(\x_{n+1})$ and on
 $i_1,\ldots,i_{n_j}$.
 Therefore,
 \begin{align*}
  1-\alpha 
  &\leq \P\left(U_{m+1} \geq U_{[\alpha]}|\X_{n+1} 
  \in A(\x_{n+1}),i_1,\ldots,i_{n_j}\right)
 \end{align*}
 The conclusion follows from the fact that
 $Y_{n+1} \in C(\X_{n+1}) \iff 
 U_{m+1} \geq U_{[1-\alpha]}$ and 
 because this holds for every sequence 
 $i_1,\ldots,i_{n_j}$.
\end{proof}

\begin{proof}[Proof of Theorem \ref{thm:convergeHPD}]
	Let $U_i := f(Y_{i}|\x_i)$, $i=1,\ldots,m$,
	$U_{n+1}:=f(Y_{n+1}|\x_{n+1})$, and
	$W:=(\x_1,\ldots,\x_m,\x_{n+1})$.
	If
	$g_{\x_i}=g_{\x_{n+1}}$ for every $i=1,\ldots,m$, then $U_1,\ldots,U_m, U_{n+1}$ are i.i.d. conditional on $W$. Indeed, for every $t \in \mathbb{R}$,
	\begin{align*}
	\P(U_i\geq t|W)&=\P(f(Y_{i}|\x_i)\geq t|\x_i) \\
	&=\P(f(Y_{n+1}|\x_{n+1})\geq t|\x_{n+1})\\
	&=\P(U_{n+1}\geq t|\x_{n+1}),
	\end{align*} 
	where the next-to-last equality follows from the definition of the profile of the density.
	
	For every $K \in \mathbb{R}$, let
	$Q(K):=\left|\{i:f(Y_i|\x_i)\geq K\} \right|$.
	Because $U_i$'s are conditionally independent and identically distributed, then
	$Q(K)|W \sim \mbox{Binomial}(m,\P(f(Y_1|\x_1)\geq K))$.
	It follows that
	$Q(K)/m \xrightarrow[a.s.]{m\longrightarrow \infty} \P(f(Y_1|\x_1)\geq K)$.
	In particular, 
	$Q(t^*)/m \xrightarrow[a.s.]{m\longrightarrow \infty} 1-\alpha$.
	Now, by definition
	$Q(T_m)/m \xrightarrow[a.s.]{m\longrightarrow \infty} 1-\alpha$. Conclude that
	$T_m \xrightarrow[a.s.]{m\longrightarrow \infty} t^*$.
	
\end{proof}

\begin{proof}[Proof of Theorem \ref{thm:equivalance}]
	Item (i) was already shown as part of the proof of Theorem \ref{thm:convergeHPD}.
	To show (ii), 
	assume that 
	$t^*(\x_a,\alpha)=t^*(\x_b,\alpha)$ for every $\alpha \in (0,1)$.
	Now,
	notice that 
	$t^*(\x_a,\alpha)$ is such that $g_{\x_a}(t^*(\x_a,\alpha))=1-\alpha.$ Conclude that 
	$g_{\x_a}(t^*(\x_a,\alpha))=g_{\x_b}(t^*(\x_b,\alpha))$  for every $\alpha \in (0,1)$. 
	Now, because $\hf$ is continuous, $\{t^*(\x_a,\alpha):\alpha \in (0,1)\}=\mbox{Im}(\hf(\cdot|\x_a))$. Thus,
	$g_{\x_a}=g_{\x_b}$, and therefore $\x_a \sim \x_b$.
\end{proof}

\end{document}